\documentclass[11pt,USenglish]{article}

\usepackage{etoolbox}

\usepackage[margin=1in]{geometry}
\usepackage[utf8]{inputenc}
\usepackage[T1]{fontenc}
\usepackage{babel}
\usepackage{lmodern} 
\usepackage{graphicx}
\usepackage{subcaption}
\usepackage{paralist}
\usepackage{microtype}

\usepackage[normalem]{ulem}
\usepackage{enumitem}

\newcommand{\titel}{Computing Tutte Paths}
\graphicspath{{./figures/}}

\usepackage{algorithm} % floating environment
\usepackage[noend]{algpseudocode} % Pseudocodes erstellen without end-constructs
\usepackage[usenames]{color}
\definecolor{hellblau}{rgb}{0.2,0.4,1} % eigene Farben definieren
\definecolor{dunkelblau}{rgb}{0,0,0.8}
\definecolor{dunkelgruen}{rgb}{0,0.5,0}
\usepackage[
	pdftex,
	colorlinks,
	linkcolor=dunkelblau,
	urlcolor=dunkelblau,
	citecolor=dunkelgruen,
	bookmarks=true,
	linktocpage=true,
	pdftitle={\titel},
	pdfauthor={},
	pdfsubject={}%,
]{hyperref}

\usepackage{amsmath}
\usepackage{amsthm} 
\usepackage{amsfonts} % Math-Befehle wie natürliche, reele Zahlen einbinden
\theoremstyle{plain} % kursiv
	\newtheorem{satz}{Satz}[] %section
	\newtheorem{theorem}[satz]{Theorem}
	\newtheorem{lemma}[satz]{Lemma}

\theoremstyle{remark} % normal
\theoremstyle{definition} % normal mit fettem Titel
	\newtheorem{definition}[satz]{Definition}
	\newtheorem{corollary}[satz]{Corollary}

\newcommand{\jset}{{\mathcal{J}}}
\newcommand{\lset}{{\mathcal{L}}}

\begin{document}
	\title{\titel}
	%\author{}
	\author{Andreas Schmid\\
	Max Planck Institute for Informatics\\
	aschmid@mpi-inf.mpg.de
	\and
	Jens M.\ Schmidt\\
	Technische Universit\"at Ilmenau\\
	jens.schmidt@tu-ilmenau.de}
	\date{}

\maketitle

\begin{abstract}
Tutte paths are one of the most successful tools for attacking Hamiltonicity problems in planar graphs. Unfortunately, results based on them are non-constructive, as their proofs inherently use an induction on overlapping subgraphs and these overlaps hinder to bound the running time to a polynomial.

For special cases however, computational results of Tutte paths are known: For 4-connected planar graphs, Tutte paths are in fact Hamiltonian paths and Chiba and Nishizeki~\cite{Chiba1989} showed how to compute such paths in linear time. For 3-connected planar graphs, Tutte paths have a more complicated structure, and it has only recently been shown that they can be computed in polynomial time~\cite{Schmid2015}.
However, Tutte paths are defined for general 2-connected planar graphs and this is what most applications need. Unfortunately, no computational results are known.

We give the first efficient algorithm that computes a Tutte path (for the general case of $2$-connected planar graphs).
One of the strongest existence results about such Tutte paths is due to Sanders~\cite{Sanders1997}, which allows to prescribe the end vertices and an intermediate edge of the desired path. Encompassing and strengthening all previous computational results on Tutte paths, we show how to compute this special Tutte path efficiently. Our method refines both, the results of Thomassen~\cite{Thomassen1983} and Sanders~\cite{Sanders1997}, and avoids overlapping subgraphs by using a novel iterative decomposition along 2-separators. Finally, we show that our algorithm runs in quadratic time.
\end{abstract}

\section{Introduction}
Among the most fundamental graph problems is the question whether a graph $G=(V,E)$ is \emph{Hamiltonian}, i.e.\ contains a cycle of length $n := |V|$. For planar and near-planar graphs, \emph{Tutte paths} have proven to be one of the most successful tools for attacking Hamiltonicity. For this reason, there is a wealth of existential results in which Tutte paths serve as main ingredient; in chronological order, these are~\cite{Tutte1977,Thomassen1983,Thomas1994,Chiba1986,Sanders1996,Sanders1997,
Thomas1997,Yu1997,Jackson2002,Thomas2005,Gao2006,Harant2009,Leighton2010,
Ozeki2014,Ozeki2014a,Kawarabayashi2015,Schmid2015,Harant2016,Brinkmann2016}.

As historical starting point to these results, Whitney~\cite{Whitney1931} proved that every $4$-connected maximal planar graph is Hamiltonian. Tutte extended this result to arbitrary $4$-connected planar graphs by showing that every 2-connected planar graph $G$ contains a Tutte path~\cite{Tutte1956,Tutte1977} (for a definition of Tutte paths, we refer to Section~\ref{sec:preliminaries}). Thomassen~\cite{Thomassen1983} proved the following generalization, which also implies that every $4$-connected planar graph is Hamiltonian-connected, i.e.\ contains a path of length $n-1$ between any two vertices. For a plane graph $G$, let $C_G$ be its outer face.

\begin{theorem}[Thomassen~\cite{Thomassen1983}] \label{thm:Thomassen}
Let $G$ be a 2-connected plane graph, $x \in V(C_G)$, $\alpha \in E(C_G)$ and $y \in V(G)-x$. Then $G$ contains a Tutte path from $x$ to $y$ through $\alpha$.
\end{theorem}

Sanders~\cite{Sanders1997} then generalized Thomassen's result by allowing to choose also the start vertex $x$ of the Tutte path arbitrarily.

\begin{theorem}[Sanders~\cite{Sanders1997}] \label{thm:Sanders}
Let $G$ be a 2-connected plane graph, $x \in V(G)$, $\alpha \in E(C_G)$ and $y \in V(G)-x$. Then $G$ contains a Tutte path from $x$ to $y$ through $\alpha$.
\end{theorem}

Apart from the above series of fundamental results, Tutte paths have been mainly used in two research branches: While the first deals with the existence of Tutte paths on graphs embeddable on higher surfaces~\cite{Thomas1994,Brunet1995,Thomas1997,Yu1997,Thomas2005,Kawarabayashi2015}, the second~\cite{Jackson1990,Gao1994,Brunet1995,Gao1995,Jackson2002,Gao2006,Nakamoto2009} investigates generalizations or specializations of Hamiltonicity such as $k$-\emph{walks}, \emph{long cycles} and \emph{Hamiltonian connectedness}.

Unfortunately, in all the results mentioned so far, very little is known about the complexity of finding a Tutte path. This is crucial, as the task of finding Tutte paths is almost always the only reason that hinders the computational tractability. The main obstruction is that Tutte paths are usually found by a decomposition of the input graph into overlapping subgraphs, on which induction is applied. Without any additional argument, these overlapping subgraphs prevent to bound the running time to be even polynomial~\cite{Gouyou-Beauchamps1982,Schmid2015}. The only known computational results on Tutte paths are~\cite{Gouyou-Beauchamps1982,Asano1984,Chiba1989,Ozeki2014,Schmid2015}. While it is known how to compute Tutte paths for planar 4-connected graphs~\cite{Chiba1989} efficiently (there Tutte paths are just Hamiltonian paths), for planar 3-connected graphs it was only recently shown that there is indeed a polynomial-time algorithm that finds Tutte paths as well as the more general 2-walks~\cite{Schmid2015}. However, for the most versatile and heavily used Tutte paths in 2-connected planar graphs, no computational result is known so far.

\paragraph{Our Results.}
Our motivation is two-fold. First, we want to make Tutte paths accessible to algorithms. We will show that Tutte paths can be computed in quadratic time.
This has impact on almost all the applications using Tutte paths listed above. 
For several of them, we immediately obtain efficient algorithms where no polynomial-time algorithms were known before.

For example, Tutte paths, as described in~\cite{Sanders1997}, were used in~\cite{Harant2016} to show that every \emph{essentially 4-connected} planar graph (i.e., a planar $3$-connected graph $G$ in which, for any $3$-separator $S$ of $G$, $G-S$ contains one component that is a single vertex) contains a cycle of length at least $\frac{n+4}{2}$ and one of length at least $\frac{3n}{5}$ if every vertex has degree at least four. 
As the existence proofs in this paper are constructive, our result directly implies a efficient (in fact, an $O(n^2)$-time) algorithm for the computation of these cycles.
In~\cite{Brinkmann2016}, it was shown that every 3-connected planar graph having at most three 3-separators is Hamiltonian. If a 3-connected planar graph contains exactly one 3-separator, one can use the algorithm given in this paper to compute a Hamiltonian cycle in $O(n^2)$ time.

The results in~\cite{Thomas1994,Sanders1996,Thomas1997,Kawarabayashi2015} use
Theorems~\ref{thm:Thomassen} and~\ref{thm:Sanders} and their authors conjecture the existence of polynomial-time algorithms for their problems, by either hinting to the constructive nature of their proofs or by highlighting similarities to~\cite{Chiba1989}.
Our algorithm provides the necessary details and proves these conjectures in the affirmative.

Second, we aim for computing the strongest possible known variant of Tutte paths, encompassing the many incremental improvements on Tutte paths made over the years. We will therefore develop an algorithm for Sander's existence result~\cite{Sanders1997}, which is in many aspects best possible. For example, Sanders~\cite{Sanders1997} showed that it is only possible to prescribe an edge if it is contained in $C_G$. His result is also known to be immediately extendable to connected planar graphs~\cite{Ozeki2014a} (the corresponding algorithmic extension can be done by simply using block-cut trees). 
Jackson et al.~\cite{Jackson2002} showed that every circuit graph contains even a Tutte \emph{cycle} through any two prescribed vertices and an edge on the outer face. However, Sander's result is still best-possible, as this cannot be expected from 2-connected graphs (as Figure~\ref{fig:Jackson_in_2-connected} shows).
For the special case of $4$-connected planar graphs, we additionally extend the description given in~\cite{Chiba1989} by removing the restriction that the endpoints of the Tutte path must lie on the outer face.

\begin{figure}[!htb]
\centering
\includegraphics[width=0.4\textwidth]{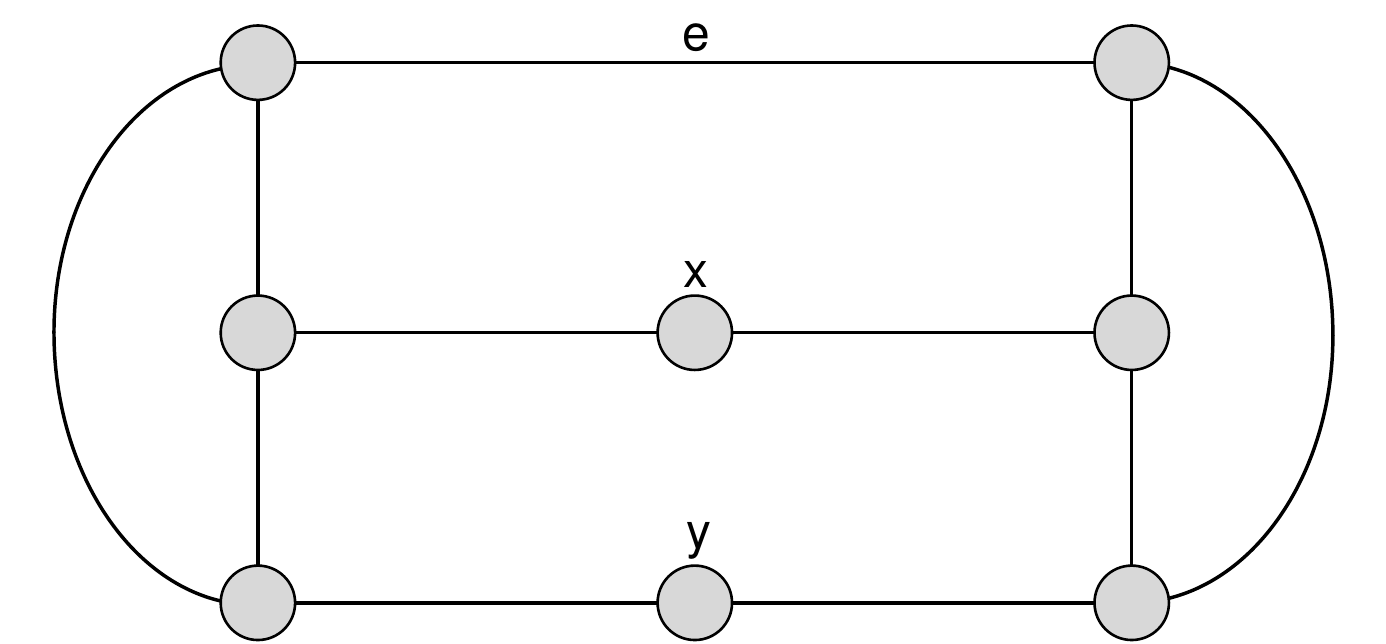}
\caption{A 2-connected planar graph that has no Tutte cycle through $x,y$ and $e$.}
\label{fig:Jackson_in_2-connected}
\end{figure}

All results in our paper will be self-contained. We will first give a decomposition that refines the ones used for Theorems~\ref{thm:Thomassen} and~\ref{thm:Sanders}, and allows to decompose $G$ into graphs that pairwise intersect in at most one edge. We then show that this small overlap does not prevent us from achieving a polynomial running time.
All graphs will be simple. We proceed by showing how this decomposition can be computed efficiently in order to find the Tutte paths of Theorem~\ref{thm:Sanders}. Our main result is hence the following:

\begin{theorem}\label{thm:main}
Let $G$ be a 2-connected plane graph, $x \in V(G)$, $\alpha \in E(C_G)$ and $y \in V(G)-x$. Then a Tutte path of $G$ from $x$ to $y$ through $\alpha$ can be computed in time $O(n^2)$.
\end{theorem}

Section~\ref{sec:decomposition} presents the non-overlapping decomposition that proves the existence of Tutte paths. On the way to our main result, we give full algorithmic counterparts of the approaches of Thomassen and Sanders; for example, we describe non-overlapping variants of Theorem~\ref{thm:Thomassen} and of the \emph{Three Edge Lemma}~\cite{Thomas1994,Sanders1996}, which was used in the purely existential result of Sanders~\cite{Sanders1997} as a black box.

\paragraph{Our Techniques.}
We follow the idea of~\cite{Chiba1989} and construct a Tutte path that is based on certain 2-separators of the graphs constructed during our decomposition. This depends on many structural properties of the given graph. In~\cite{Chiba1989}, the necessary properties however follow from the restriction to the class of internally 4-connected planar graphs, the restriction on the endpoints of the desired Tutte path, and the fact that the Tutte paths computed recursively are actually Hamiltonian.

In contrast, here we give new insights into the much wilder structure of Tutte paths of 2-connected planar graphs, allow $x,y \notin C_G$, and hence extend this technique. We show that based on the prescribed vertices and edge, there is always a set of unique non-interlacing 2-separators that are contained in every possible Tutte path of the given graph. We then use this set of 2-separators to iteratively construct one Tutte path and use this iterative procedure to avoid overlappings in the decomposition of the input graph.

\section{Preliminaries}\label{sec:preliminaries}
We assume familiarity with standard graph theoretic notations as in~\cite{Diestel2010}. Let $deg(v)$ be the degree of a vertex $v$.
We denote the subtraction of a graph $H$ from a graph $G$ by $G-H$, and the subtraction of a vertex or edge $x$ from $G$ by $G-x$.

A \emph{$k$-separator} of a graph $G=(V,E)$ is a subset $S \subseteq V$ of size $k$ such that $G-S$ is disconnected. A graph $G$ is \emph{$k$-connected} if $n > k$ and $G$ contains no $(k-1)$-separator. For a path $P$ and two vertices $x,y \in P$, let $xPy$ be the smallest subpath of $P$ that contains $x$ and $y$. For a path $P$ from $x$ to $y$, let $inner(P) := V(P)-\{x,y\}$ be the set of its inner vertices. Paths that intersect pairwise at most at their endvertices are called \emph{independent}.

A connected graph without a 1-separator is called a \emph{block}. A \emph{block of a graph} $G$ is an inclusion-wise maximal subgraph that is a block. Every block of a graph is thus either $2$-connected or has at most two vertices. It is well-known that the blocks of a graph partition its edge-set. A graph $G$ is called a \emph{chain of blocks} if it consists of blocks $B_{1},B_{2},\dots,B_{k}$ such that $V(B_i) \cap V(B_{i+1})$, $1 \leq i < k$, are pairwise distinct $1$-separators of $G$ and $G$ contains no other $1$-separator. In other words, a chain of blocks is a graph, whose block-cut tree~\cite{Harary1966} is a path.

A \emph{plane} graph is a planar embedding of a graph. Let $C$ be a cycle of a plane graph $G$. For two vertices $x,y$ of $C$, let $xCy$ be the clockwise path from $x$ to $y$ in $C$. For a vertex $x$ and an edge $e$ of $C$, let $xCe$ be the clockwise path in $C$ from $x$ to the endvertex of $e$ such that $e \notin xCe$ (define $eCx$ analogously). Let the subgraph of $G$ \emph{inside} $C$ consist of $E(C)$ and all edges that intersect the open set inside $C$ into which $C$ divides the plane. For a plane graph $G$, let $C_G$ be its outer face.

A central concept for Tutte paths is the notion of $H$-bridges (see~\cite{Tutte1977} for some of their properties): For a subgraph $H$ of a graph $G$, an \emph{$H$-bridge} of $G$ is either an edge that has both endvertices in $H$ but is not itself in $H$ or a component $K$ of $G - H$ together with all edges (and the endvertices of these edges) that join vertices of $K$ with vertices of $H$. A $H$-bridge is called \emph{trivial} if it is just one edge.
A vertex of a $H$-bridge $L$ is an \emph{attachment} of $L$ if it is in $H$, and an \emph{internal} vertex of $L$ otherwise.
An \emph{outer} $H$-bridge of $G$ is a $H$-bridge that contains an edge of $C_G$.

A \emph{Tutte path} (\emph{Tutte cycle}) of a plane graph $G$ is a path (a \emph{cycle}) $P$ of $G$ such that every outer $P$-bridge of $G$ has at most two attachments and every $P$-bridge at most three attachments. In most of the cases we consider, $G$ will be 2-connected, so that every $P$-bridge has at least two attachments. For vertices $x,y$ and an edge $\alpha \in C_G$, let an \emph{$x$-$\alpha$-$y$-path} be a Tutte path from $x$ to $y$ that contains $\alpha$. An \emph{$x$-$y$-path} is an \emph{$x$-$\alpha$-$y$-path} for an arbitrarily chosen edge $\alpha \in C_G$.

\section{Non-overlapping Decomposition}\label{sec:decomposition}

After excluding several easy cases of the decomposition, we prove Thomassen's Theorem~\ref{thm:Thomassen} constructively and then show how to use this for a proof of the Three Edge Lemma. The Three Edge Lemma, in turn, will allow for a constructive proof of Sander's Theorem~\ref{thm:Sanders} without overlapping subgraphs.
We will use induction on the number of vertices. In all proofs about Tutte paths of this chapter, the induction base is a triangle, in which the desired Tutte path can be found trivially; thus, we will assume in these proofs by induction hypothesis that graphs with less vertices contain Tutte paths. All graphs in the induction will be simple.

The following sections cover different cases of the induction steps of the three statements to prove, starting with some easy cases for which a decomposition into edge disjoint subgraphs was already given~\cite{Thomassen1983}. From now on, let $G$ be a simple plane 2-connected graph with outer face $C_G$ and let $x \in V(G)$, $\alpha \in E(C_G)$ and $y \in V(G)-x$. If $\alpha=xy$, the desired path is simply $xy$; thus, assume $\alpha \neq xy$. Since $G$ is 2-connected, $C_G$ is a cycle.

\subsection{The Easy Cases}\label{sec:instances}

We say that $G$ is \emph{decomposable into $G_L$ and $G_R$} if it contains subgraphs $G_L$ and $G_R$ such that $G_L \cup G_R=G$, $V(G_L) \cap V(G_R)=\{c,d\}$, $x \in V(G_L)$, $\alpha \in E(G_R)$, $V(G_L) \neq \{x,c,d\}$ and $V(G_R) \neq \{c,d\}$ (or the analogous setting with $y$ taking the role of $x$). 
In particular, $G_L \neq \{c,d\}$, even if $x \in \{c,d\}$. Hence $\{c,d\}$ is a 2-separator of $G$. There might exist multiple pairs $(G_L,G_R)$ into which $G$ is decomposable; we will always choose a pair that minimizes $|V(G_R)|$. Note that $G_R$ intersects $C_G$ (for example, in $\alpha$), but $G_L$ does not have to intersect $C_G$. In~\cite{Thomassen1983}, it was shown that every decomposable graph $G$ contains a Tutte path, without using recursion on overlapping subgraphs.

\begin{lemma}[\cite{Thomassen1983}]\label{lemma:2separators}
If $G$ is decomposable into $G_L$ and $G_R$, then $G$ contains an $x$-$\alpha$-$y$-path.
\end{lemma}
\begin{proof}
Let $G'_L$ and $G'_R$ be the plane graphs obtained from $G_L$ and $G_R$, respectively, by adding the edge $cd$ if this does not already exist (see Figure~\ref{fig:G_LG_R}). Let $G^*_R$ be the graph obtained from $G_R$ by subdividing $cd$ with a new vertex $z$. Clearly, each of the graphs $G_L'$, $G_R'$ and $G_R^*$ is 2-connected and contains less vertices than $G$.

\begin{figure}[!htb]
\centering
\includegraphics[width=1.0\textwidth]{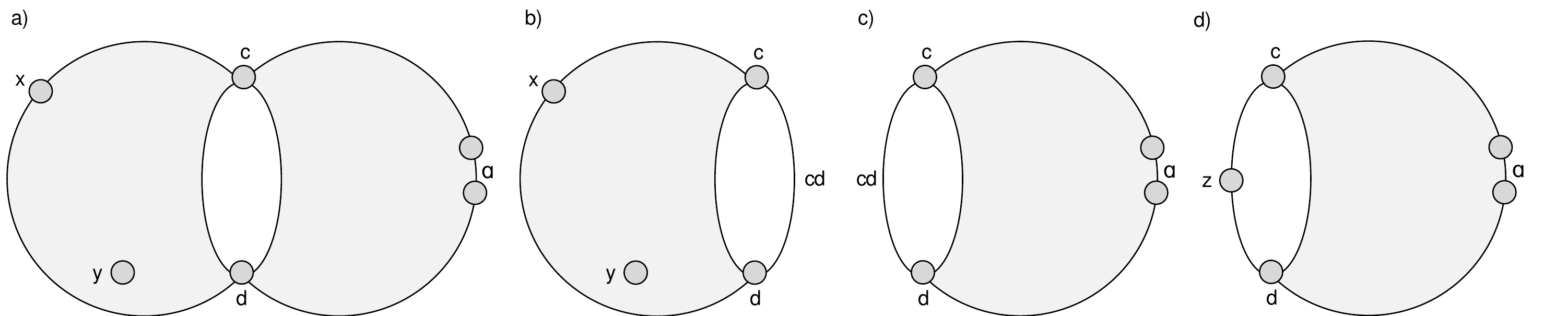}
\caption{a) shows a graph $G$ that is decomposable into $G_L$ and $G_R$. The figures b) to d) show the graphs $G_L',G_R'$ and $G_R^*$ (in this order).}
\label{fig:G_LG_R}
\end{figure}

Assume first that $y \in G_L$. By induction, $G_L'$ contains an $x$-$cd$-$y$-path $P_L$ and $G_R'$ contains a $c$-$\alpha$-$d$-path $P_R \not\ni cd$ (this requires to find a plane embedding of $G_{R'}$ whose outer face contains $\alpha$; here and later, such an embedding can always be found by stereographic projection). Then $P := (P_L - cd) \cup P_R$ is an $x$-$\alpha$-$y$-path of $G$, as $\{c,d\}$ is a 2-separator and thus every $P_L$-bridge of $G_L'$ and every $P_R$-bridge of $G_R'$ has the same attachments as its corresponding $P$-bridge of $G$.

Otherwise, $y \in G_R - \{c,d\}$. We split this case in two sub-cases. First, assume $x \in \{c,d\}$ and without loss of generalization $x=c$. By induction, $G_R'$ contains an $x$-$\alpha$-$y$-path $P_R$. Suppose $P_R$ does not contain $d$. Then $d$ is contained in a $P_R$-bridge $K$ of $G_R'$ as internal vertex and $cd \in K$. Since $cd \in C_{G_R'}$, $K$ has exactly two attachments (one of which is $x$), and these form a 2-separator implying that $G$ is decomposable into a smaller graph than $G_R$, which contradicts our choice of the decomposition. Hence, $d \in P_R$. If $cd \notin P_R$, $P_R$ is a Tutte path of $G$, as $d \in P_R$ implies that $G_L - cd$ is a $P_R$-bridge of $G$ having two attachments. If $cd \in P_R$, let $e$ be any edge in $G_L \cap C_G$; by induction, $G_L$ contains a $c$-$e$-$d$-path $P_L$. Then $P_L \cup (P_R - cd)$ is an $x$-$\alpha$-$y$-path of $G$. 

Now assume $x \notin \{c,d\}$. We will again merge two Tutte paths by induction, but have to ensure that $cd$ is not contained in any of them; to this end, we use $G_R^*$ instead of $G_R'$. By induction, there is a $z$-$\alpha$-$y$-path $P_R$ in $G_R^*$; $P_R$ contains either $zc$ or $zd$, say without loss of generalization $zc$. By the same argument as in the previous case, we have $d \in P_R$. By induction, $G_L'$ contains a $x$-$cd$-$d$-path $P_L$. Then $P := (P_L -d) \cup (P_R - z)$ is an $x$-$\alpha$-$y$-path of $G$, as $\{c,d\} = P_L \cap P_R$ and since every $P_L$- or $P_R$-bridge of $G_L$ or $G_R$, respectively, has the same attachments as its corresponding $P$-bridge of $G$.
\end{proof}

Even if $G$ is not decomposable into $G_L$ and $G_R$, $G$ may contain other $2$-separators $\{c,d\}$ that allow for a similar reduction as in Lemma~\ref{lemma:2separators} (for example, when modifying its prerequisites to satisfy $\{x,\alpha,y\} \subseteq G_R-\{c,d\}$). 

\begin{lemma}[\cite{Thomassen1983}]\label{lemma:isolated bridge}
Let $\{c,d\}$ be a 2-separator of $G$ and let $J$ be a $\{c,d\}$-bridge of $G$ having an internal vertex in $C_G$ such that $x$, $y$ and $\alpha$ are not in $J$. Then $G$ contains an $x$-$\alpha$-$y$-path.
\end{lemma} 

\begin{proof}
Let $G'$ be the plane graph obtained from $G$ by deleting all internal vertices of $J$. Since $x \notin J$, $G'$ contains at least three vertices. First, consider the case $E(C_G)-E(J) = \{\alpha\}$. Then $G'$ is 2-connected, as the 2-connectivity of $G$ and the deletion of the internal vertices of $J$ for $G'$ imply that any 1-separator $z$ of $G'$ must separate $c$ from $d$. By induction, $G'$ contains an $x$-$\alpha$-$y$-path $P$. Since $c,d \in P$ and $J$ has two attachments, $P$ is also a $x$-$\alpha$-$y$-path of $G$.

In the remaining case $E(C_G)-E(J) \neq \{\alpha\}$, we add the edge $cd$ to $G'$ where $C_G \cap J$ used to be embedded, unless $cd$ is already contained in $G'$. Clearly, $G'$ is 2-connected and $|V(G')| < n$, since $J$ contains an internal vertex. By induction, $G'$ contains an $x$-$\alpha$-$y$-path $P$. If $cd \notin P$, $cd$ is contained in a $P$-bridge of $G'$ that has two attachments and its corresponding $P$-bridge of $G$ has exactly the same attachments, so that $P$ is also a $x$-$\alpha$-$y$-path of $G$.

Now assume $cd \in P$ and let $J^* := J \cup \{cd\}$ such that $cd$ is embedded where $G-V(J)$ used to be embedded. Then $J^*$ is 2-connected and $|V(J^*)| < n$. Let $\alpha_{J^*}$ denote an arbitrary edge in $C_{J^*}-cd$. By induction, $J^*$ contains a $c$-$\alpha_{J^*}$-$d$-path $P_{J^*}$. Then the path obtained from $P$ by replacing $cd$ with $P_{J^*}$ is an $x$-$\alpha$-$y$-path of $G$, as $\{c,d\}$ separates the $P$- and $P_{J^*}$-bridges of $G$.
\end{proof}

\subsection{Proof of Theorem~\ref{thm:Thomassen}}
We now prove that $G$ contains a Tutte path from $x \in V(C_G)$ to $y \in V(G)-x$ through $\alpha \in E(C_G)$. If Lemma~\ref{lemma:2separators} or~\ref{lemma:isolated bridge} can be applied, we obtain such a Tutte path directly, so assume their prerequisites are not met. 
Let $l_\alpha$ be the endvertex of $\alpha$ that appears first when we traverse $C_G$ in clockwise order starting from $x$, and let $r_\alpha$ be the other endvertex of $\alpha$. If $y \in xC_Gl_\alpha$, we interchange $x$ and $y$ (this does not change $l_\alpha$); hence, we have $y \notin xC_Gl_\alpha$. If $y = r_\alpha$, we mirror the embedding such that $y$ becomes $l_\alpha$ and proceed as in the previous case; hence, $y \notin xC_Gr_\alpha$.

We define two paths $P$ and $Q$ in $G$, whose union will, step by step, be modified into a Tutte path of $G$. Let $Q := xC_Gl_\alpha$ and let $H := G - V(Q)$; in particular, $y \notin Q$ and, if $x$ is an endvertex of $\alpha$, $Q = \{x\}$. Since $G$ is not decomposable, we have $deg(r_\alpha) \geq 3$, as otherwise the neighborhood of $r_\alpha$ would be the 2-separator of such a decomposition. Since $deg(r_\alpha) \geq 3$, $r_\alpha$ is incident to an edge $e \notin C_G$ that shares a face with $\alpha$. Let $B_1$ be the block of $H$ that contains $e$. It is straight-forward to prove the following about $B_1$ (see Thomassen~\cite{Thomassen1983}), which shows that every vertex of $C_G$ is either in $Q$ or in $B_1$.

\begin{lemma}[\cite{Thomassen1983}] \label{lemma:there-is-no-L}
$B_1$ contains $C_G - V(Q)$ and is the only block of $H$ containing $r_\alpha$.
\end{lemma}

\begin{figure}[!htb]
\centering
\includegraphics[width=0.7\textwidth]{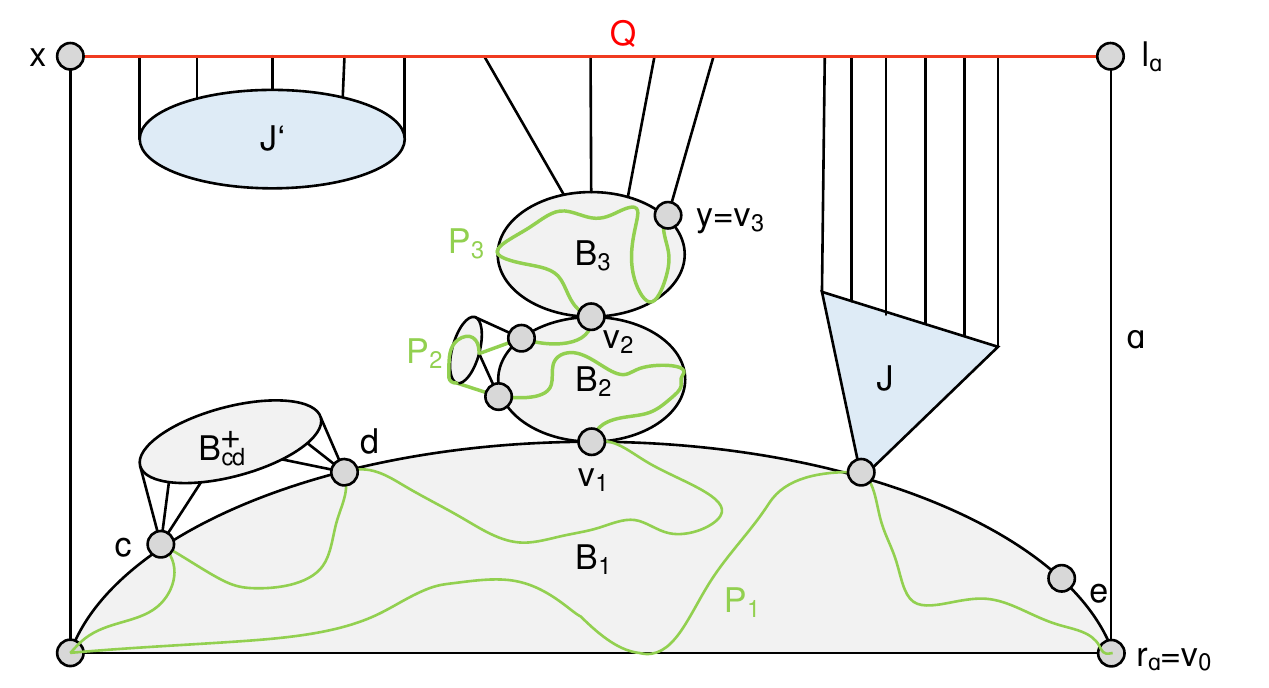}
\caption{The paths $Q$ and $P=P_1 \cup P_2 \cup P_3$, the subgraph $H$ of $G$ and its minimal chain of blocks $K = B_1,B_2,B_3$, and a $(K \cup C_G)$-bridge $J$. A $(K \cup C_G)$-bridge like $J'$ cannot exist due to Lemmas~\ref{lemma:2separators} and~\ref{lemma:isolated bridge}.}
\label{fig:HandB}
\end{figure}

Consider a component $A$ of $H$ that does not contain $B_1$. Then the neighborhood of $A$ in $G$ is in $Q$ and must contain a 2-separator of $G$ due to planarity. Hence, either $y \in A$ and we can apply Lemma~\ref{lemma:2separators} or $y \notin A$ and we can apply Lemma~\ref{lemma:isolated bridge}. Since both contradicts our assumptions, $H$ is connected and contains $B_1$ and $y$. Let $K$ be the minimal plane chain of blocks $B_1,\dots,B_l$ of $H$ that contains $B_1$ and $y$ (hence, $y \in B_l$). Let $v_i$ be the intersection of $B_i$ and $B_{i+1}$ for $1 \leq i \leq l-1$; in addition, we set $v_0:=r_\alpha$ and $v_l:=y$. 

Consider any $(K \cup C_G)$-bridge $J$. Since Lemma~\ref{lemma:isolated bridge} cannot be applied, $J$ has an attachment $v_J \in K$. Further, $J$ cannot have two attachments in $K$, as this would contradict the maximality of the blocks in $K$. Let $C(J)$ be the shortest path in $C_G$ that contains all vertices in $J \cap C_G$ and does not contain $r_\alpha$ as inner vertex (here, $r_\alpha$ serves as a \emph{reference vertex} of $C_G$ that ensures that the paths $C(J)$ are chosen consistently on $C_G$). Let $l_J$ be the endvertex of $C(J)$ whose counterclockwise incident edge in $C_G$ is not in $C(J)$ and let $r_J$ be the other endvertex of $C(J)$.

\subsubsection{Decomposing along Maximal 2-Separators}% and $\boldsymbol{\eta}(\boldsymbol{K})$}
At this point we will deviate from the original proof of Theorem~\ref{thm:Thomassen} in~\cite{Thomassen1983}, which continues with an induction on every block of $K$ that leads to overlapping subgraphs in a later step of the proof. Instead, we will show that a $v_0$-$v_l$-path $P$ of $K$ can be found iteratively without having overlapping subgraphs in the induction.

For every block $B_i \neq B_1$ of $K$, we choose an arbitrary edge $\alpha_i = l_{\alpha_i}r_{\alpha_i}$ in $C_{B_i}$. In $B_1$ we choose $\alpha_1$ such that $\alpha_1$ is incident to the endvertex of $C_{B_1} \cap C_G$ that is not $r_\alpha$.
As done for $G$, we may assume for every $B_i$ that $l_{\alpha_i}$ is the endvertex of $\alpha_i$ that is contained in $v_{i-1}C_{B_i}\alpha_i$ and that $v_i \notin v_{i-1}C_{B_i}r_{\alpha_i}$ and (by mirroring the planar embedding and interchanging $v_i$ and $v_{i-1}$ if necessary). However, unlike $G$, every $B_i$ may satisfy the prerequisites of Lemmas~\ref{lemma:2separators} and~\ref{lemma:isolated bridge}.
By induction hypothesis of Theorem~\ref{thm:Thomassen}, $B_i$ contains a $v_{i-1}$-$\alpha_i$-$v_i$-path $P_i$. In~\cite{Thomassen1983}, the outer $P_i$-bridges of $B_i$ are not only being processed during this induction step, but also in a later induction step when modifying $Q$. We avoid such overlapping subgraphs by using a new iterative structural decomposition of $B_i$ along certain 2-separators on $C_{B_i}$. This decomposition allows us to construct $P_i$ iteratively such that the outer $P_i$-bridges of $B_i$ are not part of the induction applied on $B_i$. Eventually, $P := \bigcup_{1 \leq i \leq l} P_i$ will be the desired $v_0$-$v_l$-path of $K$.

The outline is as follows. After explaining the basic split operation that is used by our decomposition, we give new insights into the structure of the Tutte paths $P_i$ of the blocks $B_i$. These are used in Section~\ref{sec:ConstructionOfP} to define the iterative decomposition of every block $B_i$ into a modified block $\eta(B_i)$, which will in turn allow to compute every $P_i$ step-by-step. This gives the first part $P$ of the desired Tutte path $x$-$\alpha$-$y$ of $G$. Subsequently, we will show how the remaining path $Q$ can be modified to obtain the second part.

For a 2-separator $\{c,d\} \subseteq C_B$ of a block $B$, let $B_{cd}^+$ be the $\{c,d\}$-bridge of $B$ that contains $cC_Bd$ and let $B_{cd}^-$ be the union of all other $\{c,d\}$-bridges of $B$ (note that $B_{cd}^+$ contains the edge $cd$ if and only if $B_{cd}^+$ is trivial); see Figure~\ref{fig:HandB}. For a 2-separator $\{c,d\} \subseteq C_B$, let \emph{splitting off} $B_{cd}^+$ (from $B$) be the operation that deletes all internal vertices of $B_{cd}^+$ from $B$ and adds the edge $cd$ if $cd$ does not already exist in $B$. Our decomposition proceeds by iteratively splitting off bridges $B_{cd}^+$ from the blocks $B_i$ of $K$ for suitable 2-separators $\{c,d\} \subseteq C_{B_i}$ (we omit the subscript $i$ in such bridges $B_{cd}^+$, as it is determined by $c$ and $d$). The following lemma restricts these 2-separators to be contained in specific parts of the outer face.

\begin{lemma}\label{lemma:bridgeattachments}
Let $P'$ be a Tutte path of a block $B$ such that $P'$ contains an edge $\alpha'$ and two vertices $a,b \in C_B$. Then every outer $P'$-bridge $J$ of $B$ has both attachments in $aC_Bb$ or both in $bC_Ba$. If additionally $J$ is non-trivial and $P' \neq \alpha'$, the attachments of $J$ form a 2-separator of $B$.
\end{lemma}
\begin{proof}
Let $e$ be an edge in $J \cap C_B$ and assume without loss of generalization that $e \in aC_Bb$. Let $c$ and $d$ be the last and first vertices of the paths $aC_Be$ and $eC_Bb$, respectively, that are contained in $P'$ (these exist, as $a$ and $b$ are in $P'$). Then $J$ has attachments $c$ and $d$ and no further attachment, as $P'$ is a Tutte path. This gives the first claim. For the second claim, let $z$ be an internal vertex of $J$. Since $P' \neq \alpha'$, $P'$ contains a third vertex $c \notin \{a,b\}$. As $c$ is not contained in $J$, $\{c,d\}$ separates $z$ and $c$ and is thus a 2-separator of $B$.
\end{proof}

For every block $B_i \neq B_l$ of $K$, let the \emph{boundary points} of $B_i$ be the vertices $v_{i-1},l_{\alpha_i},r_{\alpha_i},v_i$ and let the \emph{boundary parts} of $B_i$ be the inclusion-wise maximal paths of $C_{B_i}$ that do not contain any boundary point as inner vertex (see Figure~\ref{fig:boundary_parts}a; note that boundary parts may be single vertices). Hence, every boundary point will be contained in any possible $v_{i-1}$-$\alpha_i$-$v_i$-path $P_i$, and there are exactly four boundary parts, one of which is $\alpha_i$. 
Now, if $P_i \neq \alpha_i$, applying Lemma~\ref{lemma:bridgeattachments} for all boundary points $a,b \in \{v_{i-1},l_{\alpha_i},r_{\alpha_i},v_i\}$ and $\alpha' := \alpha_i$ implies that the two attachments of every outer non-trivial $P_i$-bridge of $B_i$ form a 2-separator that is contained in one boundary part of $B_i$. For this reason, our decomposition will split off only 2-separators that are contained in boundary parts.

\begin{figure}[!htb]
\centering
\includegraphics[width=0.6\textwidth]{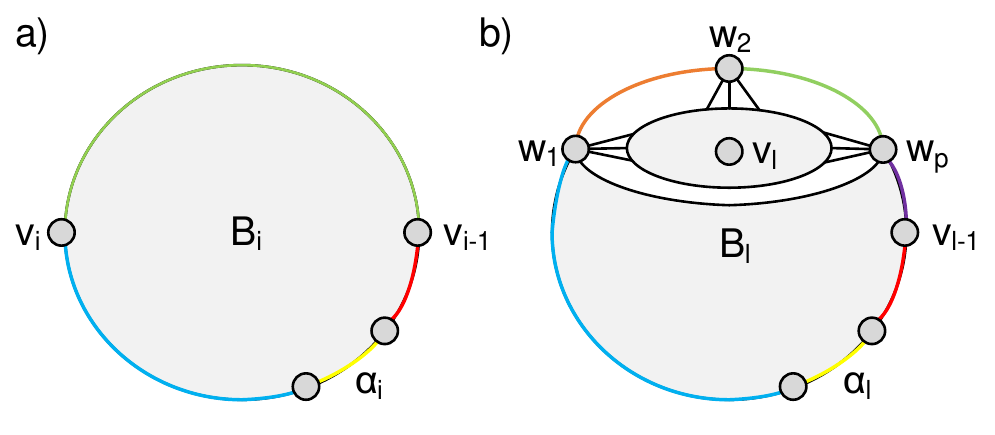}
\caption{a) The boundary points and -parts of a block $B_i \neq B_l$. b) An instance in which the block $B_l$ contains a 2-separator $\{w_1,w_p\}$ that splits off $v_l$.}
\label{fig:boundary_parts}
\end{figure}

In principle, we will do the same for the block $B_l$. If $v_l \in C_{B_l}$, we define the boundary points of $B_l$ just as before for $i < l$. However, $B_l$ is special in the sense that $v_l$ may not be in $C_{B_l}$. Then we have to ensure that we do not loose $v_l$ when splitting off a 2-separator, as $v_l$ is supposed to be contained in $P_l$ (see Figure~\ref{fig:boundary_parts}b).

To this end, consider for $v_l \notin C_{B_l}$ the 2-separator $\{w_1,w_p\} \subseteq C_{B_l}$ of $B_l$ such that $B^+_{w_1,w_p}$ contains $v_l$, the path $w_1C_{B_l}w_p$ is contained in one of the paths in $\{v_{l-1}C_{B_l}\alpha_l,\alpha_l,\alpha_lC_{B_l}v_{l-1}\}$ and $w_1C_{B_l}w_p$ is of minimal length if such a 2-separator exists.
The restriction to these three parts of the boundary is again motivated by Lemma~\ref{lemma:bridgeattachments}: If $P_l \neq \alpha_l$ and there is an outer non-trivial $P_l$-bridge of $B_l$, its two attachments are in $P_l$ and thus we only have to split off 2-separators that are in one of these three paths to avoid these $P_l$-bridges in the induction. If the 2-separator $\{w_1,w_p\}$ exists, let $w_1,\dots,w_p$ be the $p \geq 2$ attachments of the $w_1C_{B_l}w_p$-bridge of $B_l$ that contains $v_l$, in the order of appearance in $w_1C_{B_i}w_p$; otherwise, let for notational convenience $w_1 := \dots := w_p := l_{\alpha_i}$. In the case $v_l \notin C_{B_l}$, let the \emph{boundary points} of $B_l$ be $v_{l-1},l_{\alpha_l},r_{\alpha_l},w_1,\dots,w_p$ and let the \emph{boundary parts} of $B_l$ be the inclusion-wise maximal paths of $C_{B_l}$ that do not contain any boundary point as inner vertex.

\begin{lemma}\label{lemma:w1_wp}
If the $2$-separator $\{w_1,w_p\}$ exists, it is unique and every $v_{l-1}$-$\alpha_l$-$v_l$-path $P_l$ of $B_l$ contains the vertices $w_1,\dots,w_p$.
\end{lemma}
\begin{proof}
Let $J \subset B^+_{w_1,w_p}$ be the $w_1C_{B_l}w_p$-bridge of $B_l$ that contains $v_l$ and has attachments $w_1,\dots,w_p$. For the first claim, assume to the contrary that there is a 2-separator $\{w_1',w_{p'}'\} \neq \{w_1,w_p\}$ of $B_l$ having the same properties as $\{w_1,w_p\}$. By the connectivity of $J$ and the property that restricts $\{w_1',w_{p'}'\}$ to the three parts of the boundary of $B_l$, $\{w_1',w_{p'}'\}$ may only split off a subgraph containing $v_l$ if $w_1C_{B_l}w_p \subset w_1'C_{B_l}w_{p'}'$. This however contradicts the minimality of the length of $w_1'C_{B_l}w_{p'}'$.

For the second claim, let $P_l$ be any $v_{l-1}$-$\alpha_l$-$v_l$-path of $B_l$. Assume to the contrary that $w_j \not \in P_l$ for some $j \in \{1,\dots,p\}$. Then $w_j$ is an internal vertex of an outer $P_l$-bridge $J'$ of $B_l$. By Lemma~\ref{lemma:bridgeattachments}, both attachments of $J'$ are in $C_{B_l}$. However, since $J$ contains a path from $w_j \notin P_l$ to $v_j \in P_l$ in which only $w_j$ is in $C_{B_l}$, at least one attachment of $J'$ is not in $C_{B_l}$, which gives a contradiction.
\end{proof}

Lemma~\ref{lemma:w1_wp} ensures that the boundary points of any $B_i$ are contained in every Tutte path $P_i$ of $B_i$. Every block $B_i \neq B_l$ has exactly four boundary parts and $B_l$ has at least three boundary parts (three if $v_l \notin C_{B_l}$ and $\{w_1,w_p\}$ does not exist), some of which may have length zero. For every $1 \leq i \leq l$, the boundary parts of $B_i$ partition $C_{B_i}$, and one of them consists of $\alpha_i$. This implies in particular that $B_i$ has at least two boundary parts of length at least one unless $B_i = \alpha_i$.
We need some notation to break symmetries on boundary parts. For a boundary part $Z$ of a block $B$, let $\{c,d\}^* \subseteq Z$ denote two elements $c$ and $d$ (vertices or edges) such that $cC_Bd$ is contained in $Z$ (this notation orders $c$ and $d$ consistently to the clockwise orientation of $C_B$); if $cC_Bd$ is contained in some boundary part of $B$ that is not specified, we just write $\{c,d\}^* \subseteq C_B$.

We now define which 2-separators are split off in our decomposition. Let a 2-separator $\{c,d\}^* \subseteq C_B$ of $B$ be \emph{maximal in a boundary part $Z$ of $B$} if $\{c,d\} \subseteq Z$ and $Z$ does not contain a 2-separator $\{c',d'\}$ of $B$ such that $cC_Bd \subset c'C_Bd'$. Let a 2-separator $\{c,d\}^* \subseteq C_B$ of $B$ be \emph{maximal} if $\{c,d\}^*$ is maximal with respect to at least one boundary part of $B$. Hence, every maximal 2-separator is contained in a boundary part, and 2-separators that are contained in a boundary part are maximal if they are not properly ``enclosed'' by other 2-separators on the same boundary part.

Let two maximal 2-separators $\{c,d\}^*$ and $\{c',d'\}^*$ of $B$ \emph{interlace} if $\{c,d\} \cap \{c',d'\} = \emptyset$ and their vertices appear in the order $c,c',d,d'$ or $c',c,d',d$ on $C_B$ (in particular, both 2-separators are contained in the same boundary part of $B$). In general, maximal 2-separators of a block $B_i$ of $K$ may interlace; for example, consider the two maximal 2-separators when $B_i$ is a cycle on four vertices in which $v_{i-1}$ and $v_i$ are adjacent. However, the following lemma shows that such interlacing is only possible for very specific configurations.

\begin{lemma}\label{lemma:nointerlacing}
Let $\{c,d\}^*$ and $\{c',d'\}^*$ be interlacing 2-separators of $B_i$ in a boundary part $Z$ such that $c' \in cC_{B_i}d$ and at least one of them is maximal. Then $d'C_{B_i}c = v_{i-1}v_i = \alpha_i$.
\end{lemma}

\begin{proof}
Since $\{c,d\}$ is a 2-separator, $B_i-\{c,d\}$ has at least two components. We argue that there are exactly two. Otherwise, $B_i-\{c,d\}$ has a component that contains the inner vertices of a path $P'$ from $c$ to $d$ in $B_i - (C_{B_i} - \{c,d\})$. Then $B_i-\{c',d'\}$ has a component containing $(P' \cup C_{B_i})-\{c',d'\}$ and no second component, as this would contain the inner vertices of a path from $c'$ to $d'$ in $B_i - ((P' \cup C_{B_i})-\{c',d'\})$, which does not exist due to planarity. Since this contradicts that $\{c',d'\}$ is a 2-separator, we conclude that $B_i-\{c,d\}$, and by symmetry $B_i-\{c',d'\}$, have exactly two components.

By the same argument, $inner(cC_{B_i}d)$ and $inner(dC_{B_i}c)$ are contained in different components of $B_i-\{c,d\}$ and the same holds for $inner(c'C_{B_i}d')$ and $inner(d'C_{B_i}c')$ in $B_i-\{c',d'\}$. Hence, the component of $B_i-\{c,d'\}$ that contains $inner(cC_{B_i}d') \neq \emptyset$ does not intersect $inner(d'C_{B_i}c)$. If $inner(d'C_{B_i}c) \neq \emptyset$, this implies that $\{c,d'\} \subseteq Z$ is a 2-separator of $B_i$, which contradicts the maximality of $\{c,d\}$ or of $\{c',d'\}$. Hence, $inner(d'C_{B_i}c) = \emptyset$, which implies that $d'C_{B_i}c$ is an edge. As $Z$ is not an edge, $d'C_{B_i}c = \alpha_i$. Since $c$ and $d'$ are the only boundary points of $B_i$, either $\{c,d'\} = \{v_{i-1},v_i\}$ or $B_i = B_l$, $v_l \notin C_{B_l}$, $\{c,d'\} = \{v_{i-1},w_2\}$, $v_{i-1} = w_1$ and $w_2=w_p$. However, the latter case is impossible, as then $\{c,d'\}$ would be a 2-separator that separates $inner(cC_{B_i}d') \neq \emptyset$ and $v_l$, which contradicts the maximality of $\{c,d\}$ or of $\{c',d'\}$. This gives the claim.
\end{proof}

If two maximal 2-separators interlace, Lemma~\ref{lemma:nointerlacing} thus ensures that these two are the only maximal 2-separators that may contain $v_{i-1}$ and $v_i$, respectively. This gives the following direct corollary.

\begin{corollary}\label{cor:nointerlacing}
Every block of $K$ has at most two maximal 2-separators that interlace.
\end{corollary}

Note that any boundary part may nevertheless contain arbitrarily many (pairwise non-interlacing) maximal 2-separators. The next lemma strengthens Lemma~\ref{lemma:bridgeattachments}.

\begin{lemma}\label{lemma:strongerbridgeattachments}
Let $P_i$ be a $v_{i-1}$-$\alpha_i$-$v_i$-path of $B_i$. Let $J$ be a non-trivial outer $P_i$-bridge of $B_i$ and let $e$ be an edge in $J \cap C_{B_i}$. Then the attachments of $J$ are contained in the boundary part of $B_i$ that contains $e$.
\end{lemma}

\begin{proof}
Let $c$ and $d$ be the attachments of $J$ such that $e \in cC_{B_i}d$ and let $Z$ be the boundary part of $B_i$ that contains $e$. If $P_i = \alpha_i$, $v_{i-1} = l_{\alpha_i}$ and $v_i = r_{\alpha_i}$ are the only boundary points of $B_i$. Then $c$ and $d$ are the endvertices of $Z = v_iC_{B_i}v_{i-1} \ni e$, which gives the claim.

Otherwise, let $P_i \neq \alpha_i$. By applying Lemma~\ref{lemma:bridgeattachments} with $a=l_{\alpha_i}$ and $b=r_{\alpha_i}$, $\{c,d\}$ is a 2-separator of $B_i$ that is contained in $C_{B_i}$. By definition of $w_1,\dots,w_p$, there are at least three independent paths between every two of these vertices in $B_i$; thus, $\{c,d\}$ does not separate two vertices of $\{w_1,\dots,w_p\}$. Since all other possible boundary points ($v_{i-1},l_{\alpha_i},r_{\alpha_i},v_i$) are contained in $P_i$, applying Lemma~\ref{lemma:bridgeattachments} on these implies that $\{c,d\}$ does not separate two vertices of these remaining boundary points. Hence, if $\{c,d\} \not\subseteq Z$, we have $B_i = B_l$ and $v_l \notin C_{B_l}$ such that $\{c,d\}$ separates $\{w_1,\dots,w_p\}$ from the remaining boundary points. Since the $P_i$-bridge $J$ does not contain $\alpha_l \in P_i$, $cC_{B_l}d \subseteq J$ contains $\{w_1,\dots,w_p\}$, but $inner(cC_{B_l}d)$ does not contain any other boundary point. As $v_l \in P_i$, at least one of $\{w_1,w_p\}$ must be in $P_i$, say $w_p$ by symmetry. Then $d = w_p$, as $w_p \in P_i$ cannot be an internal vertex of $J$. Now, in both cases $p=2$ (which implies $c \neq w_1$, as $\{c,d\} \not\subseteq Z = w_1C_{B_l}w_2$) and $p \geq 3$, $J$ contains the edge of $P_i$ that is incident to $v_l$. As this contradicts that $J$ is a $P_i$-bridge, we conclude $\{c,d\} \subseteq Z$.
\end{proof}

Now we relate non-trivial outer $P_i$-bridges of $B_i$ to maximal 2-separators of $B_i$. In the next section, we will use this lemma as a fundamental tool for a decomposition into non-overlapping subgraphs that constructs $P$.

\begin{lemma}\label{lemma:2separators_vs_2bridges}
Let $P_i$ be a $v_{i-1}$-$\alpha_i$-$v_i$-path of $B_i$ such that $P_i \neq \alpha_i$. Then the maximal 2-separators of $B_i$ are contained in $P_i$ and do not interlace pairwise. If $J$ is a non-trivial outer $P_i$-bridge of $B_i$, there is a maximal 2-separator $\{c,d\}^*$ of $B_i$ such that $J \subseteq B_{cd}^+$.
\end{lemma}
\begin{proof}
Consider the first claim. Since $P_i \neq \alpha_i$ implies $\alpha_i \neq v_{i-1}v_i$ by contraposition, no two maximal 2-separators interlace due to Lemma~\ref{lemma:nointerlacing}. Assume to the contrary that there is a maximal 2-separator $\{c,d\}^*$ of $B_i$ such that $c$ or $d$ is not in $P_i$, say $c \notin P_i$ by symmetry (otherwise, we may flip $B_i$). Let $Z$ be the boundary part of $B_i$ that contains $\{c,d\}$. Now consider the non-trivial $P_i$-bridge $J$ of $B_i$ that contains $c$ as internal vertex. Since $c \in Z$, $J$ contains an edge of $Z$ and is thus a non-trivial outer $P_i$-bridge. Let $c'$ and $d'$ be the attachments of $J$ such that $c'C_{B_i}d' \subseteq J$. By Lemma~\ref{lemma:bridgeattachments}, $\{c',d'\}$ is a 2-separator of $B_i$. By Lemma~\ref{lemma:strongerbridgeattachments}, $\{c',d'\} \subseteq Z$. Then Lemma~\ref{lemma:nointerlacing} implies that $\{c',d'\}$ and the maximal 2-separator $\{c,d\}$ do not interlace. Since $J$ contains the incident edge of $c$ in $dC_{B_i}c$, we conclude $cC_{B_i}d \subset c'C_{B_i}d'$, which contradicts the maximality of $\{c,d\}$. This shows the first claim holds.

For the second claim, let $c'$ and $d'$ be the attachments of the given $P_i$-bridge $J$ and let $Z$ be the boundary part of $B_i$ that contains some edge $e \in J \cap C_{B_i}$. By Lemma~\ref{lemma:bridgeattachments}, $\{c',d'\}$ is a 2-separator of $B_i$. By Lemma~\ref{lemma:strongerbridgeattachments}, $\{c',d'\} \subseteq Z$. Hence, there is a maximal 2-separator $\{c,d\}^*$ of $B_i$ in $Z$ such that $\{c',d'\} \subseteq cC_{B_i}d$ and we conclude $J \subseteq B_{cd}^+$.
\end{proof}

\subsubsection{Construction of P}\label{sec:ConstructionOfP}%$\boldsymbol{\eta}(\boldsymbol{K})$}
We do not know $P_i$ in advance. However, Lemma~\ref{lemma:2separators_vs_2bridges} ensures under the condition $P_i \neq \alpha_i$ that we can split off every non-trivial outer bridge $J$ of $P_i$ by a maximal 2-separator, no matter how $P_i$ looks like. This allows us to construct $P_i$ iteratively by decomposing $B_i$ along its maximal 2-separators. Since maximal 2-separators only depend on the graph $B_i$ (in contrast to the paths $P_i$, which depend for example on the $K \cup C_G$-bridges), we can access them without knowing $P_i$ itself. We now give the details of such a decomposition.

\begin{definition}\label{def:etaK}
For every $1 \leq i \leq l$, let $\eta(B_i)$ be $\alpha_i$ if $\alpha_i = v_{i-1}v_i$ and otherwise the graph obtained from $B_i$ as follows: For every maximal 2-separator $\{c,d\}^*$ of $B_i$, split off $B_{cd}^+$. Moreover, let $\eta(K) := \eta(B_1) \cup \dots \cup \eta(B_l)$.
\end{definition}

If $\alpha_i \neq v_{i-1}v_i$, $\alpha_i$ cannot be a $v_{i-1}$-$\alpha_i$-$v_i$-path of $B_i$; hence, the maximal 2-separators of $K$ that were split in this definition do not interlace due to Lemma~\ref{lemma:2separators_vs_2bridges}. This implies that the order of the performed splits is irrelevant. In any case, we have $V(C_{\eta(B_i)}) \subseteq V(C_{B_i})$ and the only 2-separators of $\eta(B_i)$ must be contained in some boundary part of $B_i$, as there would have been another split otherwise. See Figure~\ref{fig:B_1} for an illustration of $\eta(B_l)$.

\begin{figure}[!htb]
\centering
\includegraphics[width=0.6\textwidth]{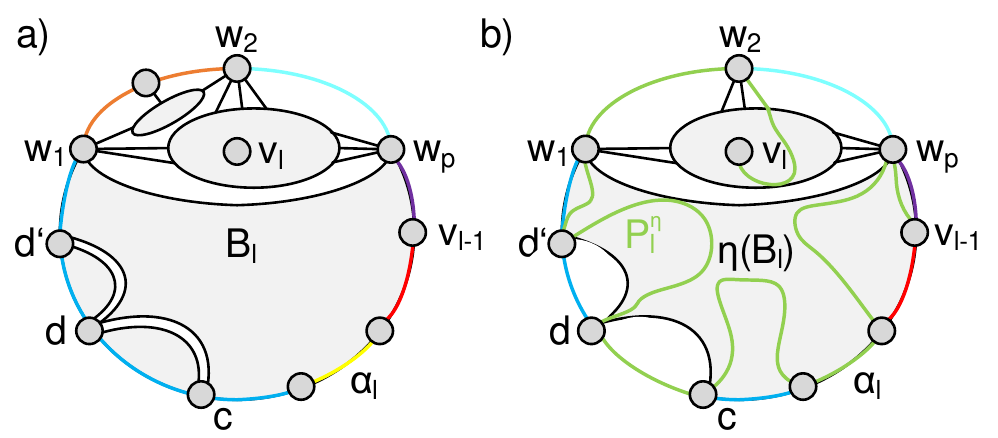}
\caption{a) A block $B_l$ with boundary points $v_{l-1},l_{\alpha_l},r_{\alpha_l},w_1,\dots,w_3$ that has two maximal 2-separators on the same boundary part. b) The graph $\eta(B_l)$.}
\label{fig:B_1}
\end{figure}

\begin{lemma}\label{lemma:eta_blocks}
Every $\eta(B_i)$ is a block. Let $P^\eta_i$ be a $v_{i-1}$-$\alpha_i$-$v_i$-path of some $\eta(B_i)$ such that $P^\eta_i \neq \alpha_i$. Then every outer $P^\eta_i$-bridge of $\eta(B_i)$ is trivial.
\end{lemma}

\begin{proof}
If $\alpha_i = v_{i-1}v_i$, $\eta(B_i) = \alpha_i$ is clearly a block. Otherwise, $B_i$ has at least three vertices and is thus 2-connected; consider two independent paths in $B_i$ between any two vertices in $\eta(B_i)$. Splitting off $B_{cd}^+$ for any maximal 2-separator $\{c,d\}^*$ (we may assume that not both independent paths are contained in $B_{cd}^+$) preserves the existence of such paths by replacing any subpath through $B_{cd}^+$ with the edge $cd$. Hence, $\eta(B_i)$ is a block.

For the second claim, we first prove that $P^\eta_i$ contains all boundary points of $B_i$. By definition, $P^\eta_i$ contains $l_{\alpha_i},r_{\alpha_i},v_{i-1}$ and $v_i$. The only possible remaining boundary points $w_1,\dots,w_p$ may occur only if $i=l$, $v_l \notin C_{B_l}$ and the 2-separator $\{w_1,w_p\}$ exists. In that case, we argue similarly as for Lemma~\ref{lemma:w1_wp}: Let $J$ be the $w_1C_{B_l}w_p$-bridge of $B_l$ that contains $v_l$; clearly, $J$ exists also in $\eta(B_l)$. Now assume to the contrary that $w_j \not \in \eta(P_l)$ for some $j \in \{1,\dots,p\}$. Then $w_j$ is an internal vertex of an outer $\eta(P_l)$-bridge $J'$ of $\eta(B_l)$. As $\eta(B_l)$ is a block, we can apply Lemma~\ref{lemma:bridgeattachments}, which implies that both attachments of $J'$ are in $C_{\eta(B_l)}$. However, since $J$ contains a path from $w_j \notin \eta(P_l)$ to $v_j \in \eta(P_l)$ in which only $w_j$ is in $C_{\eta(B_l)}$, at least one attachment of $J'$ is not in $C_{\eta(B_l)}$, which gives a contradiction.

Assume to the contrary that there is a non-trivial outer $P_i^{\eta}$-bridge $J''$ of $\eta(B_i)$ and let $c,d$ be its two attachments. Lemma~\ref{lemma:bridgeattachments} implies that $\{c,d\}$ is a 2-separator of $\eta(B_i)$ that is contained in $C_{B_i}$. If $c$ and $d$ are contained in the same boundary part of $B_i$, a supergraph of $B_{cd}^+$ would therefore have been split off for $\eta(B_i)$, which contradicts that $J''$ is non-trivial. Hence, $c$ and $d$ are contained in different boundary parts of $B_i$. Then $inner(cC_{B_i}d)$ contains a boundary point of $B_i$ and, as this boundary point is also in $P_i^{\eta}$, this contradicts that $J''$ is an outer $P_i^{\eta}$-bridge.
\end{proof}

The next lemma shows how we can construct a Tutte path $P$ of $K$ iteratively using maximal 2-separators. We will provide the details of an efficient implementation in Section~\ref{sec:computing}.

\begin{lemma}[Construction of $P$]\label{lemma:construct_P}
For every $1 \leq i \leq l$, a $v_{i-1}$-$\alpha_i$-$v_i$-path $P_i$ of $B_i$ can be constructed such that no non-trivial outer $P_i$-bridge of $B_i$ is part of an inductive call of Theorem~\ref{thm:Thomassen}.
\end{lemma}
\begin{proof}
The proof proceeds by induction on the number of vertices in $B_i$. If $B_i$ is just an edge or a triangle, the claim follows directly. For the induction step, we therefore assume that $B_i$ contains at least four vertices. If $\alpha_i = v_{i-1}v_i$, we set $P_i := \alpha_i$, so assume $\alpha_i \neq v_{i-1}v_i$. In particular, $\eta(B_i) \neq \alpha_i$ and $\alpha_i$ is no $v_{i-1}$-$\alpha_i$-$v_i$-path of $\eta(B_i)$.

As $|V(\eta(B_i))| < n$, we may apply an inductive call of Theorem~\ref{thm:Thomassen} to $\eta(B_i)$, which returns a $v_{i-1}$-$\alpha_i$-$v_i$-path $P_i^\eta \neq \alpha_i$ of $\eta(B_i)$. This does not violate the claim, since $\eta(B_i)$ does not contain any non-trivial outer $P_i^\eta$-bridge by Lemma~\ref{lemma:eta_blocks}.

Now we extend $P_i^\eta$ iteratively to the desired $v_{i-1}$-$\alpha_i$-$v_i$-path $P_i$ of $B_i$ by restoring the subgraphs that were split off along maximal 2-separators one by one. For every edge $cd \in C_{\eta(B_i)}$ such that $\{c,d\}^*$ is a maximal 2-separator of $B_i$ (in arbitrary order), we distinguish the following two cases: If $cd \notin P_i^\eta$, we do not modify $P_i^\eta$, as in $B_i$ the subgraph $B_{cd}^+$ will be a valid outer bridge.

If otherwise $cd \in P_i^\eta$, we consider the subgraph $B_{cd}^+$ of $B_i$. Clearly, $B := B_{cd}^+ \cup \{cd\}$ is a block.
Define that the \emph{boundary points} of $B$ are $c$, $d$ and the two endpoints of some arbitrary edge $\alpha_B \neq cd$ in $C_B$. This introduces the boundary parts of $B$ in the standard way, and hence defines $\eta(B)$. Note that $B$ may contain several maximal 2-separators in $cC_Bd$ that in $B_i$ were suppressed by $\{c,d\}^*$, as $\{c,d\}^*$ is not a 2-separator of $B$. In consistency with Lemma~\ref{lemma:2separators_vs_2bridges}, which ensures that no two maximal 2-separators of $B_i$ interlace, we have to ensure that no two maximal 2-separators of $B$ interlace in our case $\alpha_i \neq v_{i-1}v_i$, as otherwise $\eta(B)$ would be ill-defined. This is however implied by Lemma~\ref{lemma:nointerlacing}, as $\alpha_B \neq cd$. Since $|V(\eta(B))| < |V(B_i)|$, a $c$-$\alpha_B$-$d$-path $P_B$ of $B$ can be constructed such that no non-trivial outer $P_B$-bridge of $B$ is part of an inductive call of Theorem~\ref{thm:Thomassen}. Since $\alpha_B \neq cd$, $P_B$ does not contain $cd$. We now replace the edge $cd$ in $P_i^{\eta}$ by $P_B$. This gives the desired path $P_i$ after having restored all subgraphs $B_{cd}^+$.
\end{proof}

Applying Lemma~\ref{lemma:construct_P} on all blocks of $K$ and taking the union of the resulting paths gives $P$. In the next step, we will modify $Q$ such that $P \cup \{\alpha\} \cup Q$ becomes the desired Tutte path of $G$. By Lemma~\ref{lemma:construct_P}, no non-trivial outer $P$-bridge of $K$ was part of any inductive call of Theorem~\ref{thm:Thomassen} so far, which allows us to use these bridges inductively for the following modification of $Q$ (the existence proof in~\cite{Thomassen1983} used these arbitrarily large bridges in inductive calls for both constructing $P$ and modifying $Q$).

\subsubsection{Modification of Q}\label{sec:ModificationOfQ}
We show how to modify $Q$ such that $P \cup \{\alpha\} \cup Q$ is an $x$-$\alpha$-$y$-path of $G$. To this end, consider a $(P \cup \{\alpha\} \cup Q)$-bridge $J$ of $G$. Since Lemma~\ref{lemma:isolated bridge} cannot be applied, $J$ does not have all of its attachments in $Q$. On the other hand, if $J$ has all of its attachments in $P \subseteq K$, $J \subseteq K$ follows from the maximality of blocks and therefore $J$ satisfies all conditions for a Tutte path of $G$. Hence, it suffices to consider $(P \cup \{\alpha\} \cup Q)$-bridges that have attachments in both $P$ and $Q$. The following lemma showcases some of their properties (see Figure~\ref{fig:PQBridges}).

\begin{lemma}\label{lemma:number of attachments on P}
Let $J$ be a $(P \cup \{\alpha\} \cup Q)$-bridge of $G$ that has an attachment in $P$. Then $J \cap K$ is either exactly one vertex in $P$ or exactly one non-trivial outer $P$-bridge of $K$. In particular, $J$ has at most two attachments in $P$.
\end{lemma}

\begin{proof}
If $J$ does not contain an internal vertex of any $P$-bridge of $K$, $J$ is a $K \cup Q$-bridge of $G$. Since every such bridge has at most one attachment in $K$, $J \cap K$ is exactly this attachment, which is contained in $P$. Otherwise, let $J$ contain an internal vertex $v$ of a $P$-bridge $L$ of $K$. Then $J$ is clearly a non-trivial outer $P$-bridge and must contain $L$.

To prove the claim, we first assume to the contrary that $J$ contains an internal vertex $v'$ of a $P$-bridge $L' \neq L$ of $K$. Since the internal vertices of $J$ induce a connected graph in $G$ by definition, $J - (P \cup Q)$ contains a path from $v$ to $v'$. By the maximality of every block in $K$, this path is contained in $K-P$, which contradicts that $L$ and $L'$ are distinct. Now it remains to show that $J$ does not contain any vertex in $P-L$. Assume to the contrary that $w$ is such a vertex. Then $J - (P \cup Q)$ contains a path from $v$ to $w$ and, by the maximality of every block in $K$, this path is in $K$ and its only vertex in $P$ is $w$. This shows that $L$ has the attachment $w$, which contradicts $w \notin L$.
\end{proof}

\begin{figure}[!htb]
\centering
\includegraphics[width=0.75\textwidth]{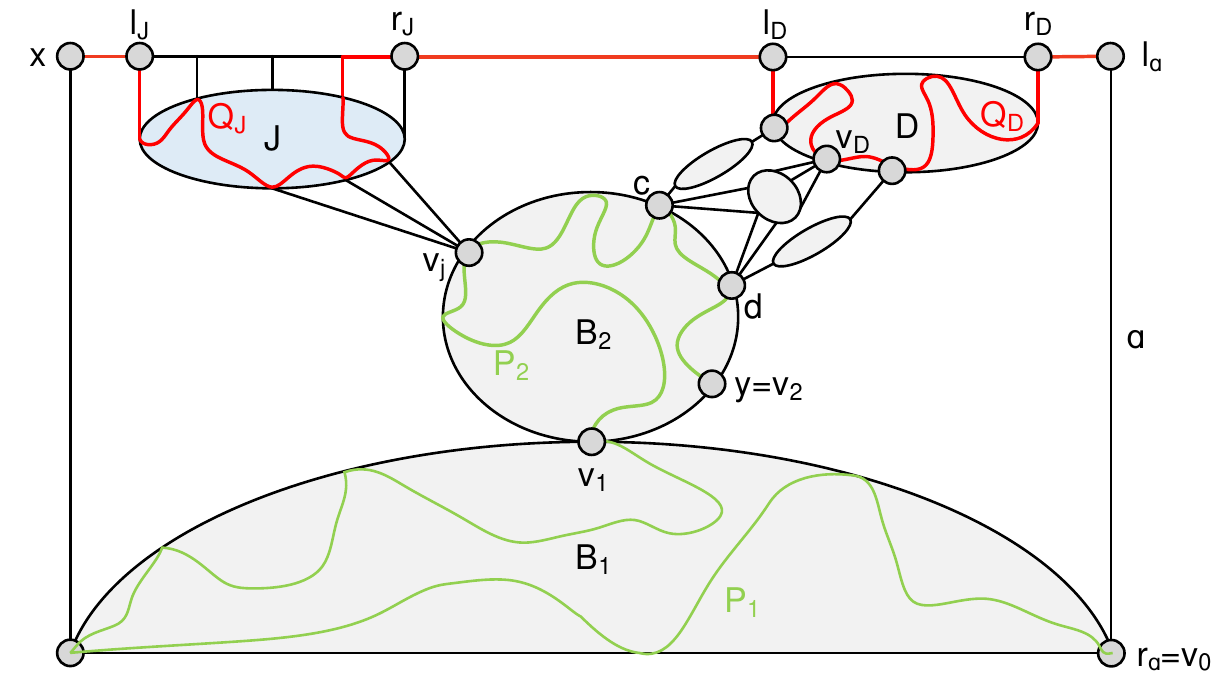}
\caption{$K=B_1 \cup B_2$ (colored gray) and two $(P \cup \{\alpha\} \cup Q)$-bridges $J$ and $D$ of $G$ such that $J$ has exactly one attachment $v_J$ in $P$ and $D$ has exactly two attachments $\{c,d\}$ in $P$ that are the attachments of a non-trivial outer $P$-bridge of $K$.}
\label{fig:PQBridges}
\end{figure}

Let $J$ be a $(P \cup \{\alpha\} \cup Q)$-bridge of $G$ that has attachments in both $P$ and $Q$ and recall that $C(J) = l_JC_Gr_J$. Because Lemma~\ref{lemma:isolated bridge} is not applicable to $G$, there is no other $(P \cup \{\alpha\} \cup Q)$-bridge than $J$ that intersects $(J \cup C(J)) - P - \{l_J,r_J\}$; in other words, $J \cup C(J)$ is everything that is enclosed by the attachments of $J$ in $G$. 
In order to obtain the Tutte path of Theorem~\ref{thm:Thomassen}, we will thus replace the subpath $C(J)$ with a path $Q_J \subseteq (J \cup C(J)) - P$ from $l_J$ to $r_J$ such that any $(Q_J \cup P)$-bridge of $G$ that intersects $(J \cup C(J)) - P - \{l_J,r_J\}$ has at most three attachments and at most two if it contains an edge of $C_G$.
Since $l_J$ and $r_J$ are contained in $Q$, no other $(P \cup \{\alpha\} \cup Q)$-bridge of $G$ than $J$ is affected by this ``local'' replacement, which proves its sufficiency for obtaining the desired Tutte path.

We next show how to obtain $Q_J$. If $C(J)$ is a single vertex, we do not need to modify $Q$ at all (hence, $Q_J := C(J)$), as then $J \cup C(J)$ does not contain an edge of $C_G$ and has at most three attachments in total (one in $Q$ and at most two in $P$ by Lemma~\ref{lemma:number of attachments on P}). If $C(J)$ is not a single vertex, we have the following lemma.

\begin{lemma}[\cite{Thomassen1983,Chiba1986}]\label{lemma:PQ-bridges}
Let $J$ be a $(P \cup \{\alpha\} \cup Q)$-bridge of $G$ that has an attachment in $P$ and at least two in $Q$. Then $(J \cup C(J)) - P$ contains a path $Q_J$ from $l_J$ to $r_J$ such that any $(Q_J \cup P)$-bridge of $G$ that intersects $(J \cup C(J)) - P - \{l_J,r_J\}$ has at most three attachments and at most two if it contains an edge of $C_G$.
\end{lemma}

\begin{proof}
By Lemma~\ref{lemma:number of attachments on P}, it suffices to distinguish two cases, namely whether $J$ has one or two attachments in $P$. Assume first that $J$ has only one attachment $v$ in $P$ (see Figure~\ref{fig:PQBridges}). Let $J' := J \cup C(J) \cup \{r_Jv\}$ (without introducing multiedges). Since we want to use induction on $J'$, we will first prove that $|V(J')| < n$ and that $J'$ is 2-connected. The first claim simply follows from $|V(K)| \geq 2$, which holds, as $r_\alpha$ and $y$ are different vertices in $K$ due to $y \notin xC_Gr_\alpha$.

For proving that $J'$ is 2-connected, consider the outer face $C_{J'}$ of $J'$ (which is not necessarily a cycle) and let $F$ be the unique inner face of $G$ that contains $v$ and $l_J$. Since $G$ is 2-connected, $F$ is a cycle and hence $vC_{J'}l_J$ is a path. Then $C_{J'} = vC_{J'}l_J \cup C(J) \cup \{r_Jv\}$, which implies that $C_{J'}$ is a cycle. Hence, for any 1-separator $w$ of $J'$, $J' - w$ has a component $S$ that does not intersect the cycle $C_{J'}$. Then the neighborhood of $S$ in $G$ is just $w$, since $J'$ and $J$ differ at most by the edge $r_Jv$. As this contradicts that $G$ is 2-connected, $J'$ is 2-connected.

By induction, $J'$ contains an $l_J$-$r_Jv$-$v$-path $Q_{J'}$. We set $Q_J := Q_{J'} - v$; then $Q_J \cap P = \emptyset$ and the neighborhood of every internal vertex of every $Q_{J'}$-bridge of $J'$ is the same in $J'$ as in $G$. Thus, every $Q_J$-bridge of $G$ corresponds to a $Q_{J'}$-bridge of $J'$, which ensures that the number of attachments of every $Q_J$-bridge of $G$ intersecting $(J \cup C(J)) - P - \{l_J,r_J\}$ is as claimed.

Assume now that $J$ has exactly two attachments $c$ and $d$ in $P$. Since $J$ is connected and contains no edge of $C(J)$, there is some cycle in $J \cup C(J)$ that contains $C(J)$. Since this cycle is also contained in $G$ and the subgraph of $G$ inside this cycle is 2-connected, $C(J)$ is contained in a block $D$ of $J \cup C(J)$ (see Figure~\ref{fig:PQBridges}). Consider a $(D \cup \{c,d\})$-bridge $L'$ of $J \cup C(J)$. Then $L'$ has at least one attachment in $D$, as otherwise $L'$ itself would be a $\{c,d\}$-bridge of $G$, which contradicts that $L'$ is contained in $J \cup C(J)$. Moreover, $L'$ has exactly one attachment in $D$, as a second attachment would contradict the maximality of $D$. By planarity, there is at most one $(D \cup \{c,d\})$-bridge $L$ that has three attachments $c$, $d$ and, say, $v_L \in D$. 

We distinguish two cases. If $L$ exists, set $v_D:=v_L$. If $L$ does not exist, let $R$ be the minimal path in $C_D - inner(C(J))$ that contains the attachments of all $(D \cup \{c,d\})$-bridges of $J$ that are in $D$. Then $R$ contains a vertex $v_D$ that splits $R$ into two paths $R_c$ and $R_d$ such that $R_c \cap R_d = \{v_D\}$. Moreover, any $(D \cup \{c,d\})$-bridge of $J$ having $c$ as one of its two attachments has its other attachment in $R_c$, and any $(D \cup \{c,d\})$-bridge of $J$ having $d$ as one of its two attachments has its other attachment in $R_d$. In either case for the vertex $v_D$, we define $\beta$ as an edge of $C_D$ that is incident to $v_D$.

As $D$ is 2-connected and contains less vertices than $G$, there is an $l_J$-$\beta$-$r_J$-path $P_D$ of $D$ by induction. Any outer $P_D$-bridge of $D$ may therefore gain either $c$ or $d$ as third attachment when considering this bridge in $G$, but not both; if $L$ exists, $L$ has still only the three attachments $\{c,d,v_L\}$ in $G$. Thus, $P_D$ is the desired path $Q_J$.
\end{proof}

By Lemma~\ref{lemma:number of attachments on P}, any $(P \cup \{\alpha\} \cup Q)$-bridge $J$ of $G$ intersects $K$ in at most one non-trivial $P$-bridge of $K$ having attachments $c$ and $d$. By Lemma~\ref{lemma:construct_P}, this non-trivial $P$-bridge was never part of an inductive call of Theorem~\ref{thm:Thomassen} before (in fact, at most its edge $cd$ was). Replacing $C(J)$ with $Q_J$ for every such $J$, as described in Lemma~\ref{lemma:PQ-bridges} and before, therefore concludes the constructive proof of Theorem~\ref{thm:Thomassen}.

\subsection{The Three Edge Lemma}
We use the constructive version of Theorem~\ref{thm:Thomassen} of the last section to deduce a construction for the following Three Edge Lemma (see~\cite{Thomas1994} and~\cite{Sanders1996} for existence proofs).

\begin{lemma}[Three Edge Lemma]\label{lemma:3-edge-lemma}
Let $G$ be a 2-connected plane graph and let $\alpha,\beta,\gamma$ be edges of $C_G$. Then $G$ has a Tutte cycle that contains $\alpha$, $\beta$ and $\gamma$.
\end{lemma}
\begin{proof}
Without loss of generality, let $\alpha$, $\beta$ and $\gamma$ appear in clockwise order on $C_G$. Let $\alpha=a'Ca$, $\beta=uCu'$ and $Q := aCu$ (see Figure~\ref{fig:three-edge-lemma}). The proof proceeds by induction on the number of vertices of $G$. In the base-case that $G$ is a triangle, the claim is satisfied, so assume $n \geq 4$.

\begin{figure}[!htb]
\centering
\includegraphics[width=0.65\textwidth]{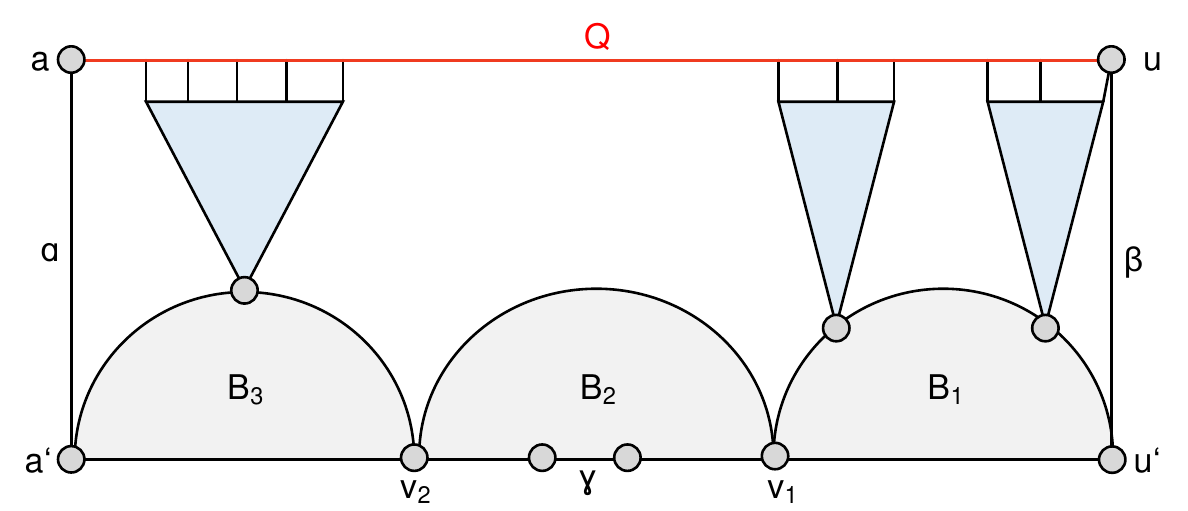}
\caption{A graph $G$ with edges $\alpha$, $\beta$, $\gamma$ that contains a plane chain of blocks $K$, as used in the Three Edge Lemma.}
\label{fig:three-edge-lemma}
\end{figure}

Let $K$ be a minimal plane chain of blocks $B_1,\ldots,B_l$ of $G - Q$ that contains $u'$ and $a'$, and let $k$ be such that $\gamma \in B_k$. Let $v_i := B_{i-1} \cap B_l$ for every $1 \leq i \leq l-1$, $v_0 := u'$ and $v_l := a'$. We define the boundary points and -parts of every $B_i$ exactly as for the blocks $B_i \neq B_l$ in the proof of Theorem~\ref{thm:Thomassen} (we set $\alpha_k = \gamma$ and, for every $i \neq k$, $\alpha_i$ to an arbitrary edge of $C_{B_i}$); note that this defines $\eta(B_i)$ for every $i$.

Now we apply Lemma~\ref{lemma:construct_P} on $G$, which constructs iteratively an $u'$-$\gamma$-$a'$-path $P$ of $K$ such that no non-trivial outer $P$-bridge of $K$ is part of an inductive call of Theorem~\ref{thm:Thomassen}. Then modifying $Q$ as described in Lemma~\ref{lemma:PQ-bridges} constructs the desired Tutte cycle $P \cup \{\alpha,\beta\} \cup Q$ of $G$.
\end{proof}

\subsection{Proof of Theorem~\ref{thm:Sanders}}

Using the Three Edge Lemma~\ref{lemma:3-edge-lemma}, we will prove Theorem~\ref{thm:Sanders} constructively by induction on the number of vertices in $G$ (again, the base case is the triangle-graph, for which the claim is easily seen to be true). For the induction step, if $x$ or $y$ is in $C_G$, the claim follows directly from Theorem~\ref{thm:Thomassen}, so we assume $x,y \notin C_G$. If there is an edge $e \in E(G)$ such that $x$ or $y$ is contained in $C_{G - e}$, we can construct an $x$-$\alpha$-$y$-path of $G-e$ (and thus of $G$) by applying Theorem~\ref{thm:Thomassen}. Thus, assume no such edge $f$ exists.

If $G$ is decomposable into $G_L$ and $G_R$, Theorem~\ref{thm:Sanders} holds by Lemma~\ref{lemma:2separators}; therefore, assume that this is not the case. In particular, there is no 2-separator in $G$ that has both vertices in $C_G$ and separates $x$ and $y$. Hence, $x$ and $y$ are in the same component of $G - C_G$. Let $K$ be the minimal plane chain of blocks $B_1,B_2,\dots,B_l$ in $G - C_G$ such that $x \in B_1$ and $y \in B_l$. Let $v_i := B_i \cap B_{i+1}$ for every $1 \leq i \leq l-1$, $v_0 := x$ and $v_l := y$.

Let $J$ be any $(K \cup C_G)$-bridge. In Theorem~\ref{thm:Thomassen}, we could choose the vertex $r_\alpha \in K \cap C_G$ as \emph{reference vertex} in order to define $C(J)$ in a consistent way. Here, the situation is more complicated, as $K$ and $C_G$ are vertex-disjoint and thus no vertex in $K \cap C_G$ exists. Instead, we take any vertex $s \in C_G$ that is contained in the same face as some vertex of $K$ as \emph{reference vertex} (not every vertex of $C_G$ may thus be $s$). Now let $C(J)$ be the shortest path in $C_G$ that contains all vertices in $J \cap C_G$ and does not contain $s$ as inner vertex. For $i \notin \{1,l\}$, we define the boundary points and -parts of $B_i$ exactly as for the blocks $B_i \neq B_l$ in the proof of Theorem~\ref{thm:Thomassen};
the boundary points $v_{l-1},l_{\alpha_l},r_{\alpha_l},v_l,w_1,\dots,w_p$ of $B_l$ and their boundary parts are defined as for $B_l$ in the proof of Theorem~\ref{thm:Thomassen} and the ones of $B_1$ (namely $v_0,l_{\alpha_1},r_{\alpha_1},v_1,z_1,\dots,z_q$) symmetric to that. 
Once we choose an edge $\alpha_i$ for every $i$, this defines $\eta(B_i)$.

\begin{figure}[!htb]
\centering
\includegraphics[width=0.9\textwidth]{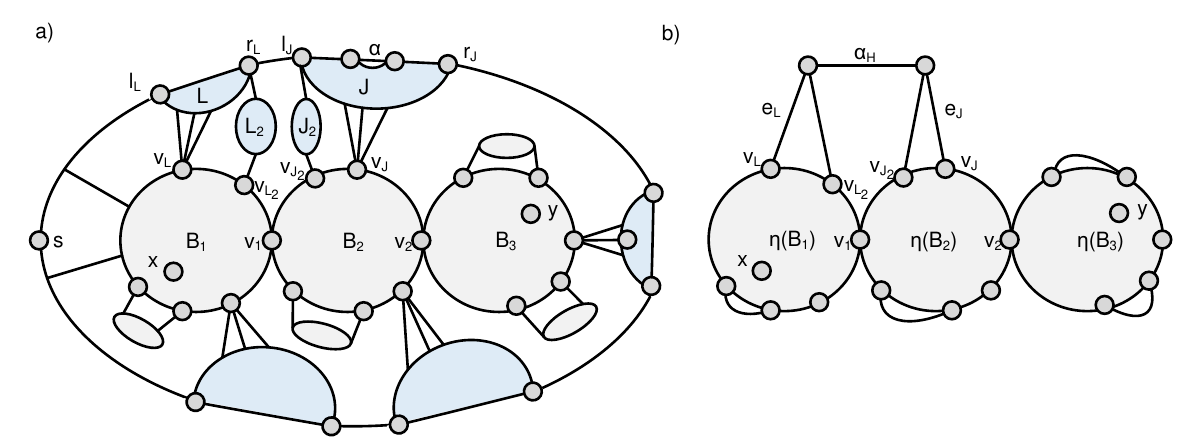}
\caption{a) Decomposing $G$ when both $x$ and $y$ are not in $C_G$. Here $K$ consists of $3$ blocks, $K_I=B_1 \cup B_2$ and $\lset,\jset$ are both of cardinality two. b) Shows the resulting $\eta(H)$ for the example in a).}
\label{fig:theorem2}
\end{figure}

By Lemma~\ref{lemma:isolated bridge}, Theorem~\ref{thm:Sanders} holds if there is a $(K \cup C_G)$-bridge of $G$ all of whose attachments are in $C_G$. Therefore, we assume further that any $(K \cup C_G)$-bridge $J$ of $G$ has exactly one attachment in $K$ and at least one attachment in $C_G$. 
Further, there are at least two $(K \cup C_G)$-bridges of $G$, as $K \cap C_G = \emptyset$ and $G$ is $2$-connected.

\begin{lemma}
No two $(K \cup C_G)$-bridges $J$ and $L$ satisfy $C_J \subseteq C_L$.  
\end{lemma}

Let $J$ be either the $(K \cup C_G)$-bridge for which $C(J)$ contains $\alpha$, or, if such a bridge does not exist, the $(K \cup C_G)$-bridge for which $l_J$ lies closest counterclockwise to $\alpha$ on $C_G$ (see Figure~\ref{fig:theorem2}). Let $L$ be the $(K \cup C_G)$-bridge for which $r_L$ lies closest counterclockwise to $l_J$ on $C_G$ (possibly $r_L = l_J$) such that $l_L \neq l_J$. Let $\jset := \{J_1,J_2,\dots,J_m\}$ be the set of all $(K \cup C_G)$-bridges $J_i$ for which $l_{J_i}=l_J$. Let $\lset:=\{L_1,L_2,\dots,L_n\}$ be the set of all $(K \cup C_G)$-bridges $L_j$ for which $r_{L_j} = r_L$ and $l_{L_j} \neq l_J$.
Then $J=J_i$ for some $i$ and, since $l_L \neq l_J$, $L \in \lset$; hence, both $\lset$ and $\jset$ are non-empty.
Let $I$ be the minimal set of consecutive indices in $\{1,\ldots,l\}$ such that $K_{I} := \bigcup_{i \in I} B_i$ contains all attachments in $K$ of the $(K \cup C_G)$-bridges in $\lset \cup \jset$. Let $f$ and $g$ denote the minimal and maximal indices of $I$.

To construct the desired $x$-$\alpha$-$y$-path of $G$, we will merge two different Tutte paths $P$ and $Q$ in edge-disjoint subgraphs of $G$. In more detail, $P$ is between the vertices $x$ and $y$ and is contained in $K \cup \bigcup_i \lset_i \cup \bigcup_j \jset_j$, while $Q$ is between $l_L$ and $r_J$ and follows $l_LC_Gr_J$, while detouring into $(K \cup C_G)$-bridges of $G$ if necessary. 

We construct $P$ by using induction on a plane change of blocks $H$ that has one block representing $K_I$ and all the $(K \cup C_G)$-bridges in $\lset \cup \jset$.
Initially, let $H$ consist of $K$ and two new artificial adjacent vertices $a$ and $b$ of degree one each. For every $L_j \in \lset$, we add an edge $e_{L_j}:=v_{L_j}a$ to $H$ (recall that $v_{L_j}$ is the unique vertex $L_j \cap K$) and for every $J_i \in \jset$, we add an edge $e_{J_i}:=v_{J_i}b$ to $H$. 
We embed $H$ into the plane by taking the embedding of $K$ from $G$ and placing $a$ and $b$ into the outer face. If $r_L \neq l_J$, we are done with the construction of $H$ and set $\alpha_H:=ab$.
Otherwise, we contract the edge $ab$ of $H$ and set $\alpha_H:=v_{J_1}b$ (note that in this case $n=1$).
In both cases, $H$ is a plane chain of blocks such that $V(B_f \cup \dots \cup B_g) \cup \{a,b\}$ is the vertex set of one of these blocks. Since $C_G$ contains at least three vertices, we have $|V(H)|<|V(G)|$. For every $i \in \{1,\dots,f-1,g+1,\dots,l\}$, let $\alpha_i$ be an arbitrary edge of $C_{B_i}$.

The crucial step in constructing a $x-\alpha_H-y$ path $P_H$ is to ensure no outer bridge of $P$ is part of any induction call. For this, we use Lemma~\ref{lemma:construct_P} on the blocks of $\eta(H)$ under the induction hypothesis of Theorem~\ref{thm:Sanders} (instead under the one of Theorem~\ref{thm:Thomassen}) 
This constructs an $x$-$\alpha_H$-$y$-path $P_H$ of $H$.

\begin{figure}[!htb]
\centering
\includegraphics[width=0.9\textwidth]{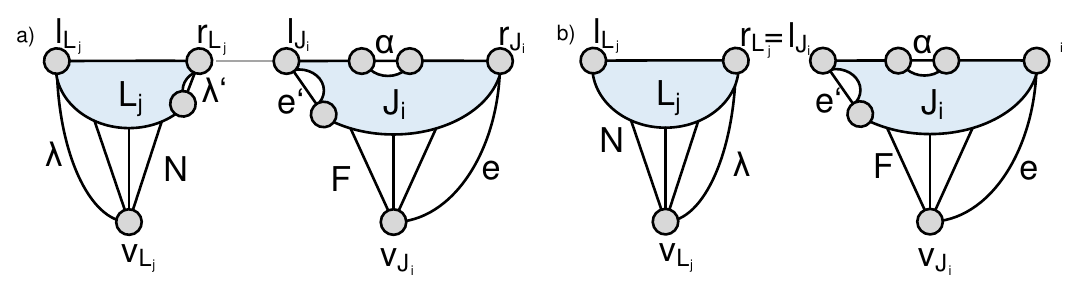}
\caption{Two examples for the subgraphs $N$ and $F$. In a) $r_{L_j}\neq l_{J_i}$, while in b) $r_{L_j} = l_{J_i}$.}
\label{fig:N_and_F}
\end{figure}

So far $P_H$ is not a subgraph of $G$, as it contains edges $v_{L_j}a$ and $v_{J_i}b$. Each of these edges represent a $(K \cup C_G)$-bridge of $G$. In the following, we show how to find Tutte paths $P_{J_i}$ and $P_{L_j}$ in the $(K \cup C_G)$-bridges $L_j$ and $J_i$. 
Note that by forcing $P_H$ through $\alpha_H$ we ensured that $P_H$ contains exactly two of these artificial edges. If $L_j$ or $J_i$ are just single edges, let $P_{J_i}:=J_i$ and $P_{L_j}:=L_j$, respectively.
If $J_i$ is not just a single edge, let $e:=v_{J_i}r_{J_i}$ and $F:=J_i \cup C(J_i) \cup \{e\}$, where $e$ is embedded such that $C(J_i)$ is part of the outer face of $F$.
Let $e'\neq e$ be an edge in $C_F$ incident to $l_J$. See Figure~\ref{fig:N_and_F} for an example.
Clearly, $F$ is 2-connected and $|V(F)| < |V(G)|$. 
If $\alpha \in E(F)$ (i.e.\ $J_i=J$), then by Lemma~\ref{lemma:3-edge-lemma} there is an Tutte cycle $P'$ that contains $e,e'$ and $\alpha$.
If $\alpha \notin J_i$, then by Theorem~\ref{thm:Thomassen} there is a $v_{J_i}$-$r_{J_i}$-path $P'$ in $F$ through $e'$.
In either case, let $P_{J_i}:=P'-e$.

It remains to show what to do if $L_j$ is not just an edge. If $r_{L_j} \neq l_{J}$, let $\lambda:=v_{L_j}l_{L_j}$ and $N:=L_j \cup C(L_j) \cup \{\lambda\}$, where $\lambda$ is embedded such that $C(L_j)$ is part of the outer face of $N$. 
Let $\lambda'$ be an incident edge to $r_{L_j}$ that is different from $\lambda$, of the outer face of $N$. Figure~\ref{fig:N_and_F} shows an example for the construction of $N$.
By Theorem~\ref{thm:Thomassen} there is a $v_{L_j}$-$\lambda'$-$l_{L_j}$-path $P_N$ of $N$. 
If otherwise $r_{L_j} = l_{J}$, then $r_{L_j}$ is already part of $P_{J_i}$ in $J_i$ and we have to ensure that we do not include it as an internal vertex of $P_N$ as well.
Let $\lambda:=v_{L_j}r_{L_j}$ and $N:=L_j \cup C(L_j) \cup \{\lambda\}$, where $\lambda$ is embedded such that $C_{L_j}$ is part of the outer face of $N$. By Theorem~\ref{thm:Thomassen}, there is a $l_{L_j}$-$\lambda$-$r_{L_j}$-path $P_N$ of $N$ and we set $P_{L_j}:=P_N-\lambda$. 
Note that if we consider the union of $P_{L_j}$ and $P_{J_i}$, then any $P_{L_j}$-bridge in $L_j$ that has $r_{L_j}$ as an attachment will also have it as an attachment in $L_j \cup J_i$.

At this point we can remove $a$ and $b$ from $P_H$, note that this disconnects $P_H$. By adding $P_{J_i}$ and $P_{L_j}$ we end up with a path $P_x$ from $x$ to $l_{L_j}$ and $P_y$ from $r_{J_i}$ to $y$. Let $Q:= r_{J_i}C_Gl_{L_j}$, to complete the proof of Theorem~\ref{thm:Sanders}, we need to modify $Q$ such that any $(P_x \cup P_y \cup Q)$-bridge of $G$ has at most three attachments and exactly two if it contains an edge of $C_G$. For this we use Lemma~\ref{lemma:PQ-bridges}, as in the proof of Theorem~\ref{thm:Thomassen} before. Note that if either $J \neq J_i$ or $L \neq L_j$, then they will become $P_x \cup P_y \cup Q$-bridges and taken care of by applying~\ref{lemma:PQ-bridges}.

\section{A Quadratic Time Algorithm}\label{sec:computing}
In this section, we give an algorithm based on the decomposition shown in Section~\ref{sec:decomposition} (see Algorithm~\ref{alg:algorithm}). It is well known that there are algorithms that compute the 2-connected components of a graph and the block-cut tree of $G$ in linear time, see~\cite{Schmidt2013a} for a very simple one. Using this on $G-Q$, we can compute the blocks $B_1,\dots,B_l$ of $K$ in time $O(n)$.

We now check if Lemma~\ref{lemma:2separators} or~\ref{lemma:isolated bridge} is applicable at least once to $G$; if so, we stop and apply the construction of either Lemma~\ref{lemma:2separators} or~\ref{lemma:isolated bridge}. Checking applicability involves the computation of special 2-separators $\{c,d\}$ of $G$ that are in $C_G$ (e.g., we did assume minimality of $|V(G_R)|$ in Lemma~\ref{lemma:2separators}). In order to find such a $\{c,d\}$ in time $O(n)$, we first compute the \emph{weak dual} $G^*$ of $G$, which is obtained from the dual of $G$ by deleting its outer face vertex, and note that such pairs $\{c,d\}$ are exactly contained in the faces that correspond to 1-separators of $G^*$. Once more, these faces can be found by the block-cut tree of $G^*$ in time $O(n)$ using the above algorithm. Since the block-cut tree is a tree, we can perform dynamic programming on all these 1-separators bottom-up the tree in linear total time, in order to find one desired $\{c,d\}$ that satisfies the respective constraints (e.g.\ minimizing $|V(G_R)|$, or separating $x$ and $\alpha$).

Now we compute $\eta(K)$. Since the boundary points of every $B_i$ are known from $K$, all \emph{maximal} 2-separators can be computed in time $O(n)$ by dynamic programming as described above. We compute in fact the nested tree structure of all 2-separators on boundary parts due to Lemma~\ref{lemma:2separators_vs_2bridges}, on which we then apply the induction described in Lemma~\ref{lemma:construct_P}. Hence, no non-trivial outer $P$-bridge of $K$ is touched in the induction, which allows to modify $Q$ along the induction of Lemma~\ref{lemma:PQ-bridges}.

\begin{algorithm}
\caption{TPATH($G,x,\alpha,y$)\Comment{method, running time without induction}}\label{alg:algorithm}
\begin{algorithmic}[1]
	\If {$G$ is a triangle or $\alpha = xy$}
		\textbf{return} the trivial $x$-$\alpha$-$y$ path of $G$\Comment{$O(1)$}
	\EndIf
	\If {Lemma~\ref{lemma:2separators} or~\ref{lemma:isolated bridge} is applicable at least once to $G$}\Comment{weak dual block-cut tree, $O(n)$}
		\State apply TPATH on $G_L$ and $G_R$ as described and \textbf{return} the resulting path\Comment{$O(1)$}
	\EndIf
	\If {there is a 2-separator $\{c,d\} \in C_G$ of $G$}
		\State do simple case 2
	\EndIf
	\State Compute the minimal plane chain $K$ of blocks of $G$ \Comment{block-cut tree of $G - Q$, $O(n)$}
	\State Compute $\eta(K)$ \Comment{dyn.\ progr. on weak dual block-cut tree, $O(n)$}
	\State Compute $P$ by the induction of Lemma~\ref{lemma:construct_P} \Comment{dyn. prog. precomputes all possible $B_{cd}^+$, $O(n)$}
	\State Modify $Q$ by the induction of Lemma~\ref{lemma:PQ-bridges} \Comment{traversing outer faces of bridges, $O(n)$}
	\State \textbf{return} $P \cup \{\alpha\} \cup Q$
\end{algorithmic}
\end{algorithm}

In our decomposition, every inductive call is invoked on a graph having less vertices than the current graph. The key insight is now to show a good bound on the total number of inductive calls to
Theorem~\ref{thm:Sanders}. In order to obtain good upper bounds, we will restrict the choice of $\alpha_i$ for every block $B_i$ of $K$ (which was almost arbitrary in the decomposition) such that $\alpha_i$ is an edge of $C_{B_i}-v_{i-1}v_i$. This prevents several situations in which the recursion stops because of the case $\alpha = xy$, which would unease the following arguments.

The next lemma shows that only $O(n)$ inductive calls are performed. Its argument is, similarly to one in~\cite{Chiba1989}, based on a subtle summation of the Tutte path differences that occur in the recursion tree.

\begin{lemma}\label{alg:analyze}
The number of inductive calls for TPATH($G,x,\alpha,y$) is at most $2n-3$.
\end{lemma}
\begin{proof}
Let $r$ be the number of inductive calls for TPATH($G,x,\alpha,y$). Let $d(i)$, $1\leq i\leq r$, be the number of smaller graphs into which we decompose the simple 2-connected plane graph of the $i$th inductive call. Let $r'$ be the number of inductive calls that satisfy $d(i)=1$. Let $t$ be the number of graphs in which we can find the desired Tutte paths trivially without having to apply induction again (i.e., triangles or graphs in which $\alpha=xy$).  

Thus, in the directed recursion tree, $t$ is the number of leaves and $r$ is the number of internal nodes, $r'$ out of which have out-degree one.
Since in a binary tree the number of internal nodes is one less than the number of leaves, the tree has at most $t-1$ internal nodes of out-degree two or more. Thus we have

\[ r\leq t-1+r'. \]

To complete the proof, we will give an upper bound for $t$ that depends on $n$. 
The $t$ instances in the leafs come in three different shapes: a triangle, a graph in which $K$ consists of only one trivial block and $Q$ can be found without applying induction (i.e., a cycle of length four) or a graph in which $\alpha=xy$. Any other instance is either decomposable into $G_L$ and $G_R$ or $K$ contains at least one non trivial block on which we have to apply induction.
If the graph in a leaf instance is just a triangle the trivially found Tutte path will be of length two and we denote the number of such leafs by $t_1$.
If a leaf represents a cycle of length four, then the trivially found Tutte path will be of length three. Let $t_2$ denote the number of such leafs.
If the graph in the leaf instance is such that $\alpha=xy$, then the Tutte path returned for this instance will be of length one. Note that this case can only appear in the root instance. This follows from the fact that we always choose $alpha$ such that $alpha \neq xy$ before we apply induction on a graph constructed in our decomposition. Thus if there is a leaf in which $alpha=xy$ then the tree consists of exactly one node and the claim is trivially true. We therefore assume that there is no such leaf from hereon.
Then there are $t=t_1+t_2$ leafs and the sum over all paths lengths in the leaves is exactly $2t_1+3t_2$. In addition a Tutte path in $G$ has length at most $n-1$.
Combining these two facts, an upper bound on $2t_1+3t_2$ can be derived by going through every internal node of the recursion tree and adding the differences between the length of the Tutte path in the current node and the sum of lengths of the Tutte paths in its children nodes to $n-1$.

If $G$ is decomposable into $G_L$ and $G_R$, then $d(i)=2$ and the Tutte path $P$ of $G$ is either $(P_L \cup P_R)-cd$ or $(P_L-d)\cup (P_R-z)$. In the first case, $P_L$ and $P_R$ intersect in $cd$ and therefore $|E(P_L)+|E(P_R)|-|E(P)|=1=d(i)-1$. In the latter case, $P_L$ contains $cd$ and $P_R$ contains one edge incident to $z$, which both will not be part of $P$; therefore, $|E(P_L)+|E(P_R)|-|E(P)|=2=(d(i)-1)+1$.

Otherwise, the graph $G$ of inductive call $i$ is decomposed along certain 2-separators and $d(i)$ depends on the number of blocks in $K$, the number of such 2-separators and the resulting $(P \cup Q)$-bridges in $G$. The following argument will also hold for inductive calls, when we apply Lemma~\ref{lemma:isolated bridge}, as the construction is similar to the case when $K$ consists of only one block and there is exactly one 2-separator in $K$.
Note that only the inductive calls on the graphs split off from $K$ increase the difference between the length of the Tutte path of $G$ and the sum off Tutte path lengths found in the children of $i$, as only in this case the graphs in the parent node and its child overlap by one edge (the decomposition shows that this is the only possible overlap).

When constructing $P$ using the induction of Lemma~\ref{lemma:construct_P}, we start with one inductive call for every block of $\eta(K)$, and every such block and every graph split off from $K$ that needs an inductive call represents another child in the recursion tree. 
Initially, $P$ is a Tutte path in $\eta(K)$ formed by the union of the Tutte paths $P_1^{\eta},\dots,P_l^{\eta},$ found in $\eta(B_1),\dots,\eta((B_l)$, where $l$ is the number of blocks in $K$. As $P_{j}$ and $P_{j+1}$, $1 \leq j \leq l-1$, do only intersect in one of their endpoints, the difference in $\sum_{j=1}^l|E(P_j)|$ and $|E(P=P_1 \cup \dots \cup P_l)|$ is zero.
For every graph that creates a child $j$ that is split off from $K$, we remove one edge from $P$ and replace it with a Tutte path $P_j$ of $j$. As $P$ and $P_j$ do not intersect in any edge, $|E(P)|+|E(P_j)|-|E(P\cup P_j)|=1$.
Thus, the difference between the length of the Tutte path computed in $i$ and the sum of lengths of Tutte paths computed in its children nodes is equal to the number $k$ of graphs we split of from $K$ and apply induction on. As $k \leq d(i)-1$ the difference therefore is at most $d(i)-1$ in this case.

If $d(i)=1$, then the Tutte path found in the child note must be at least one edge shorter than the Tutte path in the parent node. Combining all of these differences shows that the total length of paths found in the $t$ leaves is at most

\begin{align*}
2t_1+3t_2 & \leq n-1 +\sum_{1\leq i \leq r} (d(i)-1) + I -r'=n-1 + r+t-1-r+I -r' \\
2t + t_2 & \leq n+t+I-r'-2,
\end{align*}

where $I$ is the number of inductive calls on graphs that are decomposable into $G_L$ and $G_R$. This implies that

\begin{align*}
t+t_2 & \leq n+I-r'-2 \\
t & \leq n+I-r'-t_2-2 \leq n-r'+I-2
\end{align*}

Plugging this into the previous upper bound for $r$, we get $r\leq n+I-3$. Note that no 2-separator can be used in more than one inductive call that decomposes the graph into $G_L$ and $G_R$. Therefore, we obtain $I\leq n$ which concludes $r \leq 2n-3$.
\end{proof}

Hence, Algorithm~\ref{alg:algorithm} has overall running time $O(n^2)$, which proves our main Theorem~\ref{thm:main}.
We obtain as well the following direct corollary of the Three Edge Lemma~\ref{lemma:3-edge-lemma}.

\begin{corollary}\label{lemma:alg3-edge-lemma}
Let $G$ be a 2-connected plane graph and let $\alpha,\beta,\gamma$ be edges of $C_G$. Then a Tutte cycle of $G$ that contains $\alpha$, $\beta$ and $\gamma$ can be computed in time $O(n^2)$.
\end{corollary}

\bibliographystyle{abbrv}
\bibliography{Tuttepaths}

\end{document}